\def\year{2021}\relax
\documentclass[letterpaper]{article} %
\usepackage{aaai21}  %
\usepackage{times}  %
\usepackage{helvet} %
\usepackage{courier}  %
\usepackage[hyphens]{url}  %
\usepackage{graphicx} %
\usepackage{natbib}  %
\usepackage{caption} %
\usepackage{amsmath, amsfonts, amsthm}
\usepackage{xspace}
\usepackage{bbm}
\usepackage{xcolor}
\usepackage{comment}
\usepackage{bm}

\usepackage{algorithm} 
\usepackage{algpseudocode}

\DeclareMathOperator*{\argmax}{arg\,max}
\DeclareMathOperator*{\argmin}{arg\,min}

\newtheorem{proposition}{Proposition}
\newtheorem{corollary}{Corollary}
\newtheorem{lemma}{Lemma}
\newcommand{\E}{\mathbb{E}}

\newcommand{\idone}[1]{\delta\!\!\left[ \vphantom{\Big(} {#1} \right]}

\newcommand{\epsilonVar}{{{\epsilon}}}

\newcommand{\optrv}{ { {r} } }
\newcommand{\groundTruth}{ r_\text{GT}  } 
\newcommand{\bestGuess}{ r_\star  } 
\newcommand{\groundTruthM}{ M_\star  } 
\newcommand{\informationGain}[3]{ \mathbb{I}({#1} \,;\, {#2} 
\if\relax\detokenize{#3}\relax\else
\; \vert \; {#3}
\fi 
)}
\newcommand{\functional}[3]{ {#1}\left(
{#2} 
\if\relax\detokenize{#3}\relax\else
\; \vert \; {#3}
\fi
 \right)}
\newcommand{\entropy}[2]{\functional{\mathbb{H}}{#1}{#2}}
\newcommand{\expect}[3]{\mathbb{E}_{#1}\left[{#2}
\if\relax\detokenize{#3}\relax\else
\; \vert \; {#3}
\fi 
\right]}

\title{Estimating $\alpha$-Rank by Maximizing Information Gain}
\author{
Tabish Rashid\footnote{Work done during an internship at MSR Cambridge.},\textsuperscript{\rm 1}
Cheng Zhang,\textsuperscript{\rm 2}
Kamil Ciosek\textsuperscript{\rm 2}\\
}
\affiliations {
\textsuperscript{\rm 1}University of Oxford\\
\textsuperscript{\rm 2}Microsoft Research, Cambridge, UK \\
tabish.rashid@cs.ox.ac.uk,
\{cheng.zhang, kamil.ciosek\}@microsoft.com
}

\begin{document}

\maketitle

\begin{abstract}
    Game theory has been increasingly applied in settings where the game is not known outright, but has to be estimated by sampling.
For example, meta-games that arise in multi-agent evaluation can only be accessed by running a succession of expensive experiments that may involve simultaneous deployment of several agents.
In this paper, we focus on $\alpha$-rank, a popular game-theoretic solution concept designed to perform well in such scenarios.
We aim to estimate the $\alpha$-rank of the game using as few samples as possible.
Our algorithm maximizes information gain between an epistemic belief over the $\alpha$-ranks and the observed payoff. This approach has two main benefits. 
First, it allows us to focus our sampling on the entries that matter the most for identifying the $\alpha$-rank. Second, the Bayesian formulation provides a facility to build in modeling assumptions by using a prior over game payoffs.
We show the benefits of using information gain as compared to the confidence interval criterion of ResponseGraphUCB \citep{rowland2019multiagent}, 
and provide theoretical results justifying our method.

\end{abstract}

\section{Introduction}

Traditionally, game theory is applied in situations where the game is fully known. 
More recently, empirical game theory addresses the setting where this is not the case, instead, the game is initially unknown and has to be interacted with by sampling \citep{wellman2006methods}.
One area in which this is becoming increasingly common is the ranking of trained agents relative to one another. Specifically, in the field of Reinforcement Learning game-theoretic rankings are used 
not just as a metric for measuring algorithmic progress \citep{balduzzi2018re}, but also as an integral component of many population-based training methods \citep{Muller2020A, lanctot2017unified, vinyals2019grandmaster}. In particular, for ranking, two popular solution concepts have recently emerged:   
Nash averaging \citep{balduzzi2018re, nash1951non} and $\alpha$-rank \citep{omidshafiei2019alpha}. 

In this paper, we aim to estimate the $\alpha$-rank of a game using as few samples as possible.
We use the $\alpha$-rank solution concept for two reasons. First, it admits a unique solution whose computation easily scales to $K$-player games. Second, unlike older schemes such as Elo \citep{elo1978rating}, $\alpha$-rank is designed with intransitive interactions in mind. 
Because measuring payoffs can be very expensive, it is important to do it by using as few samples as possible.
For example, playing a match of chess \citep{silver2017mastering}
can take roughly 40 minutes (assuming a typical game-length of 40 and up to 1 minute per move as used during evaluation), and playing a full game of Dota 2 can take up to 2 hours \citep{berner2019dota}.
\emph{Our objective is thus to accurately estimate the $\alpha$-rank using a small number of payoff queries.}

\citet{rowland2019multiagent} proposed ResponseGraphUCB (RG-UCB) for this purpose, inspired by pure exploration bandit literature. RG-UCB maintains confidence intervals over payoffs. When they don't overlap, it draws a conclusion about their ordering, until all comparisons relevant for the computation of $\alpha$-rank have been made. While this is provably sufficient to determine the true $\alpha$-rank with a high probability in the infinite-$\alpha$ regime, their approach has two important limitations. 
First, since the frequentist criterion is indirect, relying on payoff ordering rather than the $\alpha$-rank, the obtained payoffs aren't always used optimally. Second, it is nontrivial to include useful domain knowledge about the entries or structure of the payoff matrix. 
    
To remedy these problems, we propose a Bayesian approach. Specifically, we utilize a Gaussian Process to maintain an epistemic belief over the entries of the payoff matrix, providing a powerful framework in which to supply domain knowledge. This payoff distribution induces an epistemic belief over $\alpha$-ranks.
We determine which payoff to sample by maximizing information gain between the $\alpha$-rank belief and the obtained payoff. This allows us to focus our sampling on the entries that are expected to have the largest effect on our belief over possible $\alpha$-ranks. 

\textbf{Contributions:} Theoretically, we justify the use of information gain by showing a regret bound for a version of our criterion in the infinite-$\alpha$ regime. Empirically, our contribution is threefold. First, we compare to RG-UCB on stylized games, showing that maximizing information gain provides competitive performance by focusing on sampling the more relevant payoffs. 
Second, we evaluate another objective based on minimizing the Wasserstein divergence, which offers competitive performance while being computationally much cheaper. Finally, we demonstrate the benefit of building in prior assumptions.

\section{Background}
\label{sec:background}
\paragraph{Games and $\alpha$-Rank} 
A game with $K$ players, each of whom can play $S_k$ strategies is characterized by its expected payoffs $M \in \mathcal{R}^{S_1 \times ... \times S_K}$ \citep{fudenberg1991game}.
Letting $\mathcal{S} = S_1 \times ... \times S_K$ be the space of pure strategy profiles, the game also specifies a distribution over the payoffs associated with each player when $s \in \mathcal{S}$ is played.
The $\alpha$-rank of a game is computed by first defining an irreducible Markov Chain whose nodes are pure strategy profiles in $\mathcal{S}$. We denote the stochastic matrix defining this chain as $C$. The transition probabilities of the chain $C$ are calculated as follows:
Let $\sigma, \tau \in \mathcal{S}$ be such that $\tau$ only differs from $\sigma$ in a single player's strategy and let $\eta = (\sum_{k=1}^K (|S_k| - 1))^{-1}$ be the reciprocal of the total number of those distinct $\tau$.
Let $M_k(\sigma)$ denote the expected payoff for player $k$ when $\sigma$ is played.
Then, the probability of transitioning from $\sigma$ to $\tau$ which varies only in player $k$'s strategy is
\begin{gather*}
    C_{\sigma, \tau} =
    \begin{cases}
    \eta \frac{1 - \exp(-\alpha(M_k(\tau) - M_k(\sigma)))}{1 - \exp(-\alpha m(M_k(\tau) - M_k(\sigma)))} &  \text{if } \scriptstyle  M_k(\tau) \neq 
    M_k(\sigma), \\
    \frac{\eta}{m} & \text{otherwise}.
    \end{cases}
\end{gather*}
$C_{\sigma, \upsilon} = 0$ for all $\upsilon$ that differ from $\sigma$ in more than a single player's strategy, $C_{\sigma, \sigma} = 1 - \sum_{\tau \neq \sigma} C_{\sigma, \tau}$ to ensure a valid transition distribution, and
$\alpha \geq 0, m \in \mathbb{N}^{>0}$ are parameters of the algorithm.
We define the $\alpha$-rank $r \in \mathcal{R}^{|\mathcal{S}|}$  as the unique stationary distribution of the chain $C$ \citep{omidshafiei2019alpha, rowland2019multiagent} as $\alpha \to \infty$.
In practice, a large finite value of $\alpha$ is used, or a perturbed version of the transition matrix $C$ is used with an infinite $\alpha$ to ensure the resulting Markov Chain $C$ is irreducible.

\paragraph{Single Population $\bm{\alpha}$-Rank}
In this paper we focus on the infinite-$\alpha$ regime and restrict our attention to the 2-player single population case of $\alpha$-rank which differs slightly from above. 
Importantly, our method can be easily applied to multiple populations as described above in a straightforward way, but we focus on the single population case for simplicity.
Let $S = S_1$ and $M(\sigma, \tau)$ denote the payoff when the first player plays $\sigma$ and the second player plays $\tau$. Note that $S_1 = S_2$ since the single population case considers a player playing a game against an identical player.
In this particular setting, the $\alpha$-rank $r \in \mathcal{R}^{|S|}$ and the perturbed transition matrix $C \in \mathcal{R}^{S \times S}$ is calculated as follows:
\begin{gather*}
\small
C_{\sigma, \tau} = 
\begin{cases}
\textstyle
(|S| - 1)^{-1} (1 - \epsilon)  & \text{if } M(\tau, \sigma) > M(\sigma, \tau), \\
\textstyle
(|S| - 1)^{-1} \epsilon        & \text{if } M(\tau, \sigma) < M(\sigma, \tau), \\
\textstyle
0.5 (|S| - 1)^{-1}             & \text{if } M(\tau, \sigma) = M(\sigma, \tau),
\end{cases}
\end{gather*}
for $\sigma \neq \tau$. $C_{\sigma, \sigma} = 1 - \sum_{\tau \neq \sigma} C_{\sigma, \tau}$ again to ensure a valid transition distribution and $\epsilon$ is a small perturbation to ensure irreducibillity of the resulting chain. 
We abstract the above computation into the $\alpha$-rank function $f:\mathcal{M} \rightarrow \mathcal{R}^{|S|}$, where $\mathcal{M}$ is the space of 2-player payoff matrices with $S$ strategies for each player. 

\paragraph{Wasserstein Divergence} 
Let $p$ and $q$ be probability distributions supported on $\mathcal{X}$, and $c:\mathcal{X} \times \mathcal{X} \rightarrow [0, \infty)$ be a distance.
Define $\Pi$ as the space of all joint probability distributions with marginals $p$ and $q$. Wasserstein divergence \citep{villani2008optimal} with cost function $c$, is defined as:
$$ \textstyle
\mathcal{W}_c(p,q) := \min_{\pi \in \Pi} \int_{\mathcal{X} \times \mathcal{X}} c(x,y) d\pi(x,y).
$$
In this paper, we will utilize the Wasserstein distance between our belief distributions over $\alpha$-rank, and so we set
$\mathcal{X} = \Delta^{S-1}$, the $(S-1)$ probability simplex, and use $c(x,y) = \frac12 \| x - y\|_1$, i.e. the total variation distance.
We will drop the suffix and denote this simply as $\mathcal{W}$.

\section{Related Work}
There are many methods related to the ranking and evaluation of agents in games.
Elo \citep{elo1978rating} and TrueSkill \citep{herbrich2007trueskill, minka2018trueskill} both quantify the performance of an agent using a single number,
which means they are unable to model  \textit{intransitive} interactions.
\citet{chen2016modeling} extend TrueSkill to better model such interactions, while \citet{balduzzi2018re} do the same for Elo, improving its predictive power by introducing additional parameters. 
\citet{balduzzi2018re} also re-examines the use of Nash equilibrium, proposing to disambiguate across possible equillibria by picking the one with maximum entropy. However, it is well known that computing the Nash equilibrium is computationally difficult \citep{daskalakis2009complexity} and only computationally tractable for restricted classes of games.
In this paper, we focus on $\alpha$-rank \citep{omidshafiei2019alpha} since it has been designed with \textit{intransitive} interactions in mind, it is computationally tractable for $N$-player games and shows considerable promise as a component of self-play frameworks \citep{Muller2020A}.

Empirical Game Theory \citep{wellman2006methods} is concerned with situations in which a game can only be interacted with through sampling. The most related work to ours investigates sampling strategies and concentration inequalities for the Nash equilibrium as opposed to the $\alpha$-rank.
\citet{walsh2003choosing} introduce Heuristic Payoff Tables (HPTs) in order to choose the samples that provide the most information about the currently chosen Nash equilibrium, where information is quantified as the reduction in estimated error. 
This differs from our approach both in the use of $\alpha$-rank as opposed to the Nash equilibrium as our solution concept, and in the criterion used to select the observed payoff. \citet{tuyls2020bounds} provide concentration bounds for estimated Nash equilibria. \citet{jordan2008searching} find Nash equilibria from limited data by using information gain on distributions over strategies, a concept different from our information gain on distributions over ranks. We also utilize $\alpha$-rank as the solution concept, rather than Nash equilibria.

\citet{Muller2020A} utilise $\alpha$-rank as part of a PSRO \citep{lanctot2017unified} framework. They do not use an adaptive sampling strategy for deciding which entries to sample, but are a natural application for applying our algorithm (and RG-UCB).
\citet{yang2019alpha} introduce an approximate gradient-based algorithm which does not require access to the entire payoff matrix at once in order to compute $\alpha$-rank.
Although their method does not require the entire payoff matrix at every iteration, it is not designed for operating in the same incomplete information setting that we explore in this paper since they assume every entry can be cheaply queried with no noise.
\citet{srinivas2009gaussian} prove regret bounds for Bayesian optimization with GPs. We use their concentration result to derive our bounds as well as as inspiration for our information gain criterion.

\subsection{ResponseGraphUCB}

Closest to our work is ResponseGraphUCB (RG-UCB) introduced by \citet{rowland2019multiagent}, which can be viewed as a frequentist analogue to our method which also operates in the infinite-$\alpha$ regime.
RG-UCB first specifies an error threshold $\delta > 0$ and then samples payoffs until a stopping criteria determines the estimated $\alpha$-rank is correct with probability at least $1 - \delta$.
A key observation that RG-UCB relies on, is that in the infinite-$\alpha$ regime only the ordering between relevant payoffs is important. 
e.g. For pure strategy profiles $\sigma$ and $\tau$ (with payoffs $M_k(\sigma)$ and $M_k(\tau)$ respectively) that are used in the computation of the Markov Chain transition probabilities, determining whether $M_k(\sigma) > M_k(\tau)$ or $M_k(\sigma) < M_k(\tau)$ is enough to know the transition probability accurately (their magnitude difference $|M_k(\sigma) - M_k(\tau)|$ is unimportant).
RG-UCB maintains $(1-\delta)$ confidence intervals for all values of $M_k(\sigma)$, and determines the ordering between $\sigma$ and $\tau$ is correct when they do not overlap.
A strategy profile is chosen to be sampled until all of its ordering are correctly determined.
When all orderings are correctly determined the algorithm terminates.

Since the confidence intervals are constructed using frequentist concentration inequalities, we refer to RG-UCB as being a frequentist algorithm.
In contrast, our Bayesian perspective provides a principled method for incorporating prior knowledge into our algorithm whereas it is much more difficult to encode modelling assumptions and prior knowledge with RG-UCB.
The second important difference between our work and RG-UCB is that our information gain criterion is a direct objective, which selects the payoffs to sample based on how likely the received sample is to affect the $\alpha$-rank. On the other hand, RG-UCB works indirectly, reducing uncertainty about the orderings between individual payoffs without considering their impact on the final $\alpha$-rank, which makes it less efficient.
\citet{rowland2019multiagent} also theoretically justify the use of RG-UCB in the infinite-$\alpha$ regime by proving sample complexity results, whereas we provide asymptotic regret bounds for our approach which are commonly used to justify the sample efficiency of a Bayesian algorithm \citep{srinivas2009gaussian}. 
\citet{rowland2019multiagent} additionally provide a method for obtaining uncertainty estimates in the infinite-$\alpha$ regime, which is, however, not used as part of an adaptive sampling strategy.
\section{Method}
\label{method-section}
On a high level, our method works by maintaining an epistemic belief over $\alpha$-ranks and selecting payoffs that lead to the maximum reduction in the entropy of that belief. 
\begin{figure}[t]
    \centering
    \includegraphics[width=1\linewidth]{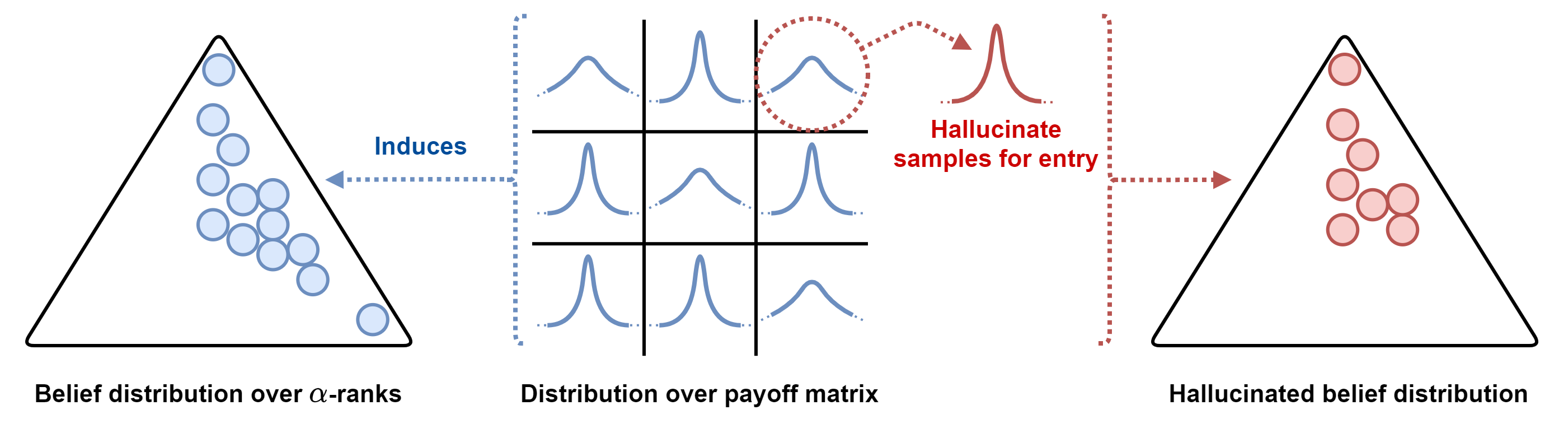}
    \caption{\textbf{Overview:} On the left, a belief over $\alpha$-ranks is induced by a belief over the payoff matrix shown in the middle. A hallucinated belief distribution is shown on the right.
    }
    \label{fig:high_level}
\end{figure}
Figure \ref{fig:high_level} provides a pictorial overview. In the middle of the figure, we maintain an explicit distribution over the entries of the payoff matrix. This payoff distribution induces a belief over $\alpha$-ranks, shown on the left. When deciding which payoff to sample, we examine hypothetical belief states after sampling, striving to end up with a belief  with the lowest entropy. One such hypothetical, or  `hallucinated' belief is shown on the right. We now describe our method formally, first describing the probabilistic model and then the  implementation.

\paragraph{Payoffs: Ground Truth and Belief}
We denote the unknown true payoff matrix as $\groundTruthM$. 
To quantify our uncertainty about what this true payoff is, we employ a Gaussian Process $M$, which also allows us to encode prior knowledge about payoff dependencies.
Our framework is sufficiently general to allow for other approaches such as Bayesian Matrix Factorization \citep{salakhutdinov2008bayesian} or probabilistic methods for Collaborate Filtering \citep{su2009survey} to be used.
We choose to use Gaussian Processes due to their flexibility in encoding prior knowledge and modelling assumptions, and their ubiquity throughout literature.

The GP models noise in the payoffs as $\tilde{M} = M + \epsilonVar$, where $\epsilonVar \sim \mathcal{N}(0, I \sigma^2_A)$.
When interacting with the game sequentially, the received payoffs are assumed to be generated as $m_t = \groundTruthM(a_t) + \epsilonVar'_t $. Here, $\epsilonVar'_t$ are i.i.d. random variables with support on the interval $[-\sigma_A, \sigma_A]$. While it may at first seem surprising that we use Gaussian observation noise in the GP model, while assuming a truncated observation noise for the actual observation, this does not in fact affect our theoretical guarantees. We provide more details in Section \ref{sec-theory}.
We denote the history of interactions at time $t$ by $H_t$. Because of randomness in the observations, $H_t$ is a random variable. The sequence of random variables $H_1, H_2, \dots$ forms a filtration. We use the symbol $h_t$ to denote particular realization of history so that $h_t = a_1,m_1,\dots,a_{t-1},m_{t-1}$.

\textbf{Belief over $\alpha$-ranks}
Our explicit distribution over the entries of the payoff matrix $M$ induces an implicit belief distribution over the $\alpha$-ranks. For all valid $\alpha$-ranks $r$, $P(r) = P(M \in f^{-1}(r))$ where $f^{-1}$ denotes the pre-image of $r$ under $f$. In other words, the probability assigned to an $\alpha$-rank $r$ is the probability assigned to its pre-image by our belief over the payoffs. Since $r$ is represented implicitly, we cannot query its mass function directly. Instead, we access $r$ via sampling. This is done by  first drawing a payoff from $m \sim M$ and then computing the resulting $\alpha$-rank $f(m)$.

\paragraph{Picking Payoffs to Query}
At time $t$, we query the payoff that provides us with the largest information gain about the $\alpha$-rank. Formally,
\footnotesize
\begin{align}
    &a_t = \argmax_a \informationGain{r}{(\tilde{M}_t(a),a)}{H_t = h_t} \nonumber \\
    &= \argmax_a \entropy{r}{H_t = h_t} \nonumber\\ & - \mathop{\E}_{\tilde{m}_t \sim \tilde{M}_t(a)} \left[\entropy{r}{H_t = h_t, A_t=a, \tilde{M}_t(a)=\tilde{m}_t}\right] \label{eq-ig-entropy} \\
    &= \argmin_a\mathop{\E}_{\tilde{m}_t \sim \tilde{M}_t(a)} \left [\entropy{r}{H_t = h_t, A_t=a, \tilde{M}_t(a)=\tilde{m}_t} \right].
    \label{eq:method_ig}
\end{align}
\normalsize
In Equation \eqref{eq-ig-entropy}, $\entropy{r}{H_t = h_t}$ is the entropy of our current belief distribution over $\alpha$-ranks, which does not depend on $a$ and can be dropped from the maximization, producing Equation \eqref{eq:method_ig}. The expectation in \eqref{eq:method_ig} has an intuitive interpretation as the expected negative entropy of our \textit{hallucinated} belief, i.e.\ belief obtained by conditioning on a sample $\tilde{m}_t$ from the current model. In essence, we are pretending to receive a sample for entry $a$, and then computing what our resulting belief over $\alpha$-ranks will be.
By picking the entry as in \eqref{eq:method_ig}, we are picking the entry whose sample will lead to the largest reduction in the entropy of our belief over $\alpha$-ranks in expectation.

\paragraph{Implementation}
\begin{algorithm}
    \small
	\caption{$\alpha$IG algorithm. $\alpha$IG(NSB) and $\alpha$IG(Bin) variants differ in entropy estimator (Line \ref{alg-line-entropy}).} 
	\begin{algorithmic}[1]
		\For {$t=1,2,\ldots T$} \label{alg-t}
		\For {$a=1,2,\ldots |\mathcal{S}|$} \label{alg-actions}
				\For {$i=1,2,\ldots N_e$}
				\State 
				\label{alg-line-halluc}
				$\tilde{m}_t \sim \tilde{M}_t(a)$
				\Comment {`Hallucinate' a payoff.}
				\State Obtain hallucinated posterior payoff:  $$P(\hat{M}_t \vert H_t = h_t, A_t=a, \tilde{M}_t(a)=\tilde{m}_t)$$ \vspace{-10pt}
			    \label{alg-line-posterior}
			    \State $D = \{r_1,\dots, r_{N_b}\}$, where $r_i \sim f(\hat{M}_t)$ i.i.d.
			    \label{alg-line-rank-distrib}
		        \State $\hat{h}^i_a = $ \textsc{estimate-entropy}( $ D $ )
		        \label{alg-line-entropy}
		        \EndFor
	        \State
	        $\hat{h}_a = \frac{1}{N_e} \sum_{i=1}^{N_e} \hat{h}^i_a $ 
	        \label{alg-line-average-entropy}
        \EndFor
		\State Query payoff  $a_t = \argmin_a \hat{h}_a$ \label{alg-line-select} \Comment Implements Eq. \eqref{eq:method_ig}.
	\EndFor
	\end{algorithmic}
	\label{alg-main}
\end{algorithm}
Our algorithm, which we refer to as $\alpha$IG, is summarized in Algorithm \ref{alg-main}. 
At a high-level, $\alpha$IG selects an action/payoff to query at each timestep (Line \ref{alg-t}).
In order to select a payoff to query as in Equation \eqref{eq:method_ig}, we must approximate the expectation for each payoff (Line \ref{alg-actions}).
In Line \ref{alg-line-halluc}, we use our epistemic model to obtain a `hallucinated' outcome $\tilde{m}_t$, as if we received a sample from selecting payoff $a$ at timestep $t$.
In Line \ref{alg-line-posterior}, we condition our epistemic model on this `hallucinated' sample $\tilde{m}_t$ in order to obtain our `hallucinated' posterior over payoffs $\hat{M}_t$.
In Line \ref{alg-line-entropy}, we empirically estimate the entropy of the resulting induced belief distribution over $\alpha$-ranks.
To approximate the expectation in \eqref{eq:method_ig}, we average out entropy estimates obtained from $N_e$ different possible hallucinated payoffs in Line \ref{alg-line-average-entropy}.
Finally, in Line \ref{alg-line-select}, we use these estimates to perform query selection as in \eqref{eq:method_ig} to select a payoff to query at timestep $t$. 

Our algorithm depends on an entropy estimator \textsc{estimate-entropy}, used in Line \ref{alg-line-entropy}. We present results for 2 different entropy estimators: simple binning and NSB. The simple binning estimator estimates the entropy using a histogram. 
For comparison, we also used NSB \citep{nemenman2002entropy}, an entropy estimator designed  to produce better estimates in the small-data regime.

\paragraph{Computational Requirements}
The main computational bottleneck of our algorithm is the calculations of $\alpha$-rank in Line \ref{alg-line-rank-distrib} of Algorithm \ref{alg-main}.
In order to perform query selection as in \eqref{eq:method_ig}, we must compute the $\alpha$-rank $|\mathcal{S}| \times N_e \times N_b$ times.
For our experiments on the 4x4 Gaussian game this results in $16 \times 10 \times 500 = 80,000$ computations of $\alpha$-rank (setting $N_e=10, N_b=500$), to select a payoff to query.
Relative to ResponseGraphUCB, our method thus requires significantly more computation in order to select a payoff to query. 
However, in Empirical Game Theory, it is commonly assumed that obtaining samples from the game is very computationally expensive (which is true in many potential practical applications \citep{berner2019dota, silver2017mastering, vinyals2019grandmaster}). 
The increased computation required by our method to select a payoff to sample should then have a negligible impact to the overall computation time required, but the increased sample efficiency could potentially lead to large speed-ups. 

We perform two simple optimizations when deploying the algorithm in practice.  
To save computational cost, we observe the same payoff $N_r$ times in Line \ref{alg-line-select}  rather than once, similar to rarely-switching bandits \citep{linear-bandits}.
Moreover, the number of samples $N_b$ we can use to estimate the entropy is limited due to the computational cost of computing $\alpha$-rank.
In order to obtain better differentiation between the entropy of beliefs arising from sampling different payoffs, we heuristically perform conditioning in Line \ref{alg-line-posterior} $N_{c}$ times. 
See Appendix \ref{sec:cond_payoff} for a more detailed discussion on this.

\section{Query Selection by Maximizing Wasserstein Divergence}
While the query objective proposed in \eqref{eq:method_ig} is backed both by an appealing intuition and a theoretical argument (see Section \ref{sec-theory}), it can be expensive to evaluate due to the cost of accurate entropy estimation. To address this difficulty, we also investigate an alternative involving the Wasserstein distance. The objective we consider is
\footnotesize
\begin{align}
    \argmax_a \E_{\tilde{m}_t \sim \tilde{M}_t} [&\mathcal{W}(P(r|H_t = h_t),\nonumber  \\ 
    &P(r|H_t = h_t, A_t=a, \tilde{M}_t(a)=\tilde{m}_t))].
    \label{eq:wasserstein}
\end{align}
\normalsize
Since the computation of Wasserstein distance from empirical distributions can be achieved by solving a linear program \citep{bonneel2011displacement}, Equation \eqref{eq:wasserstein} naturally lends itself to being approximated via samples. In our implementation, we use POT \citep{flamary2017pot} to approximate this distance.

\begin{figure}
    \centering
    \includegraphics[width=0.35\textwidth]{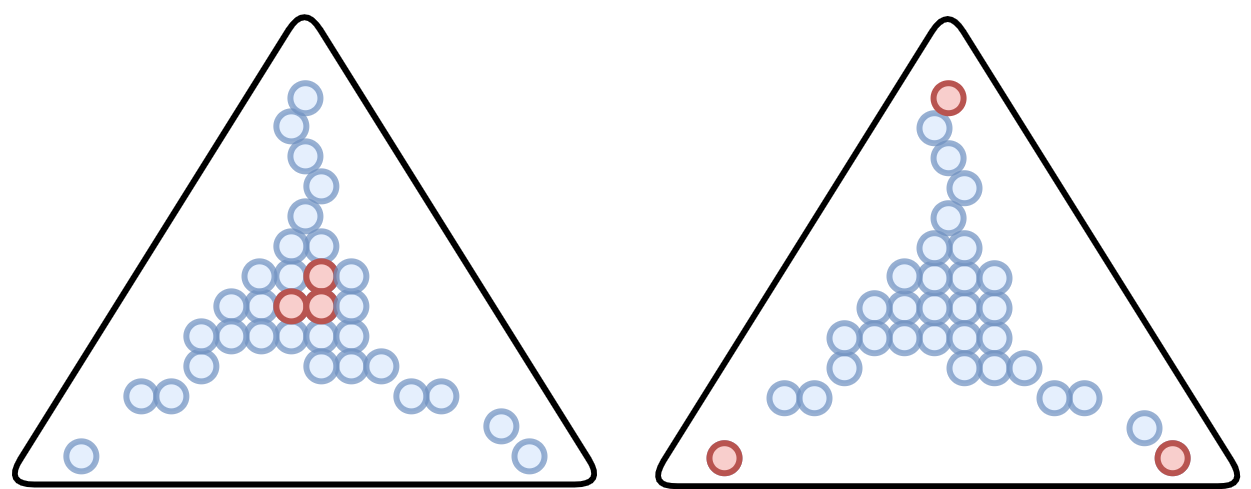}
    \caption{Diagram depicting the current belief (Blue) and 2 different hallucinated beliefs (Red). 
    We are assuming a discrete distribution over $\alpha$-ranks, where the belief is uniform across the relevant circles.}
    \label{fig:entropy_vs_wass_cartoon}
\end{figure}

The Wasserstein distance is built on the notion of cost, which allows a practitioner the opportunity to supply additional prior knowledge.
In our case, since $\alpha$-ranks are probability distributions, a natural way to measure accuracy is to use the total variation distance, which corresponds to setting the cost to $c(x,y) = \frac12 \| x - y\|_1 $. On the other hand, in cases where we are interested in finding the relative ordering of agents under the $\alpha$-rank, an alternative cost such as the Kendall Tau metric \citep{fagin2006comparing} could be used. While we emphasize the ability of the Wasserstein divergence to work with any cost, we leave the empirical study of non-standard costs for future work.

It is important to note that the objective in \eqref{eq:wasserstein} is qualitatively different to the information gain objective proposed in \eqref{eq:method_ig}.
Figure \ref{fig:entropy_vs_wass_cartoon} provides a diagram illustrating a major difference between the two objectives. 
The entropy for both belief distributions shown in red is the same.
In contrast, the Wasserstein distance in \eqref{eq:wasserstein} between the current belief in blue and the hallucinated belief in red is much smaller for the distribution on the left compared to the distribution on the right.
\section{Theoretical Results}
\label{sec-theory}
\paragraph{Notions of Regret} 
We quantify the performance of our method by measuring regret. Our main analysis relies on Bayesian regret \citep{information-directed-informs}, defined as
\begin{gather}
    \label{eq-regret-b}
    J_t^B = 1 - \expect{h_t}{P(\optrv = \bestGuess \vert H_t = h_t) }{}, 
\end{gather}
where we used $\bestGuess$ to denote the $\alpha$-rank with the highest probability under $r$ at time $t$.
In \eqref{eq-regret-b}, the expectation is over realizations of the observation model.
Since $J_t^B$, like all purely Bayesian notions, does not involve the ground truth payoff, we need to justify its practical relevance. We do this by benchmarking it against two notions of frequentist regret. 
The first measures how accurate the probability we assign to the ground truth $\groundTruth = f(\groundTruthM)$ is
\begin{gather}
    J_t^F = 1 - \expect{h_t}{P(\optrv = \groundTruth \vert H_t = h_t) }{}.
\end{gather}
The second measures if the mean of our payoff belief, which we denote $M_\mu$, evaluates to the correct $\alpha$-rank
\begin{gather}
    J_t^M = 1 - \expect{h_t}{\idone{ f(M_\mu) = \groundTruth}}{},
\end{gather}
where the symbol $\idone{\text{predicate}}$ evaluates to 1 or 0 depending on whether the predicate is true or false. In Section \ref{sec:results}, we empirically conclude that these three notions of regret are closely coupled in practice, changing at a comparable rate.

\paragraph{Regret Bounds}
As an intermediate step before discussing information gain on the $\alpha$-ranks, we first analyze the behavior of a query selection rule which maximizes information gain over the payoffs. 
\footnotesize
\begin{gather}
    \pi_{\text{IGM}}(a \vert H_t = h_t) = \argmax_a \informationGain{\tilde{M}_t}{(\tilde{M}_t(a),a)}{H_t = h_t}. 
    \label{eq-policy-ig-m}
\end{gather}
\normalsize
The following result shows that using sampling strategy $\pi_{\text{IGM}}$ for $T$ timesteps leads to a decay in regret of at least $ T e^{\mathcal{O}(-\sqrt[3]{\Delta^2 T})}$, proving it will incur no regret as $T \to \infty$.

\begin{proposition}[Regret Bound For Information Gain on Payoffs]
\label{proposition-regret-bound-payoffs}  
If we select actions using strategy $\pi_{ \text{IGM} }$, the regret at timestep $T$ is bounded as
\footnotesize
\begin{gather}
J^B_T \leq J^F_T = 1 - \expect{h_T}{P(\optrv = \groundTruth \vert H_T = h_T) }{} \leq 
        { T e^{g(T)} }  \\\;\; \text{where} \;\; 
        g(T) = \mathcal{O}(-\sqrt[3]{\Delta^2 T} ) \nonumber.
\end{gather}
\normalsize
\end{proposition}
The proof, and an explicit form of $g$ are found in supplementary material. We now proceed to our second result, where we maximize information gain on the $\alpha$-ranks directly. Consider a querying strategy that is an extension of \eqref{eq-ig-entropy} to $T$-step look-ahead, defined as
\footnotesize
\begin{gather}
    \pi_{\text{IGR}} = \argmax_{a_1,\dots,a_T} \informationGain{r}{(\tilde{M}_1(a_1),a_1), \dots, (\tilde{M}_T(a_T),a_T)}{}. 
    \label{eq-policy-ig-r}
\end{gather}
\normalsize
We quantify regret achieved by $\pi_{\text{IGR}}$ in the proposition below.

\begin{proposition}[Regret Bound For Information Gain on Belief over $\alpha$-Ranks]
If we select actions using strategy $\pi_{ \text{IGR} }$, regret is bounded as
\begin{gather*}
J^B_T = 1 - P(\optrv = \bestGuess \vert H_T = h_T) \to 0 ~~ \text{as}~~ T \to \infty.
\end{gather*}
\label{proposition-regret-bound-ranks-lim}
\end{proposition}
Proposition \ref{proposition-regret-bound-ranks-lim} provides a theoretical justification for querying the strategies that maximize information gain on the $\alpha$-ranks.
A more explicit regret bound (similar to Proposition \ref{proposition-regret-bound-payoffs}) and the proof are provided in Appendix \ref{sec:proofs}.
In practice, to avoid the combinatorial expense of selecting action sequences using $\pi_{ \text{IGR} }$, we use the greedy query selection strategy in equation \eqref{eq-ig-entropy}. While the regret result above does not carry over, this idealized setting at least provides some justification for information gain as a query selection criterion.
\section{Experiments}
\label{sec:results}
In this section, we describe our results on synthetic games, graphing the Bayesian regret $J^B_t$ described in Section \ref{sec-theory}. We also justify the use of Bayesian regret, showing that it is highly coupled with the ground truth payoff. We benchmark two versions of our algorithms, $\alpha$IG (Bins) and $\alpha$IG (NSB),  which differ in the employed entropy estimator. We compare to three baselines: \textbf{RG-UCB}, a frequentist bandit algorithm \citep{rowland2019multiagent}, \textbf{Payoff}, which maximizes the information gain about the payoff distribution, and \textbf{Uniform}, which selects payoffs uniformly at random.
\textbf{RG-UCB} represents the current SOTA in this domain, \textbf{Payoff} represents the performance of a Bayesian method that does not take into account the structure of the mapping between payoffs and $\alpha$-ranks, and \textbf{Uniform} provides a point of reference as the simplest/most naive method\footnote{We do not include Uniform on the regret graphs, since there is no reasonable value we could compute for it.}.
A detailed explanation of the experimental setup\footnote{Code is available at github.com/microsoft/InfoGainalpharank.} and details on the used hyperparameters are included in Appendix \ref{sec:experimental_setup}.

\definecolor{mygreen}{RGB}{0,0,0}
\definecolor{myred}{RGB}{0,0,0}
\definecolor{mypurple}{RGB}{0,0,0}

\paragraph{Good-Bad Games}
To investigate our algorithm, we study two environments whose payoffs are shown in Figure \ref{fig:xgood_ybad_colour_payoffs}.
We start with the relatively simple environment with 4 agents.
Figure \ref{fig:xgood_ybad_colour_payoffs} (Left) shows the expected payoffs, which we can interpret as the win-rate.
Samples are drawn from a Bernoulli distribution with the appropriate mean.
We refer to the environment as `2 Good, 2 Bad' since agents 1 and 2 are much stronger than the other 2 agents, winning $100\%$ of the games against them.  Since the ordering between agents 3 and 4 has no effect on the $\alpha$-rank, gathering samples to determine this ordering (highlighted in \textcolor{mypurple}{Purple}) does not affect the belief distribution over $\alpha$-ranks.
\begin{figure}
    \centering
    \includegraphics[height=2.2cm]{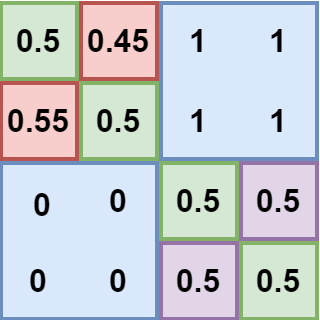}
    \hspace{15pt}
    \includegraphics[height=2.2cm]{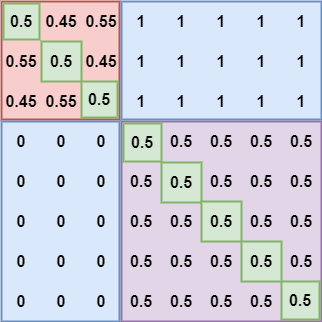}
    \caption{Payoff matrices for 2 Good, 2 Bad (Left) and 3 Good, 5 Bad (Right). Best viewed in color.}
    \label{fig:xgood_ybad_colour_payoffs}
\end{figure}
Furthermore, since we treat this as a 1-population game, the entries highlighted in \textcolor{mygreen}{Green} where each agent plays against themselves do not affect the $\alpha$-rank.
Entries that are necessary to determine the ordering between agents 1 and 2 are the most relevant for the $\alpha$-rank and are highlighted in \textcolor{myred}{Red}. 
Since agent 2 is slightly better than agent 1, the true $\alpha$-rank is $(0,1,0,0)$.
However, it can be difficult to determine the correct ordering between agents 1 and 2 without drawing many samples from these entries. The game thus provides a model for the common scenario of agents with clustered performance ratings.

\begin{figure*}
    \centering
    \includegraphics[width=0.21\textwidth]{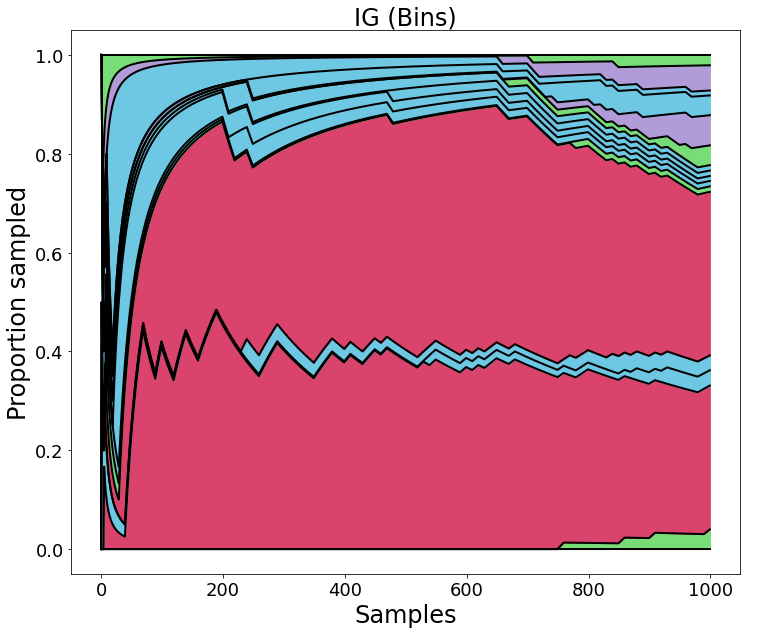}
    \includegraphics[width=0.21\textwidth]{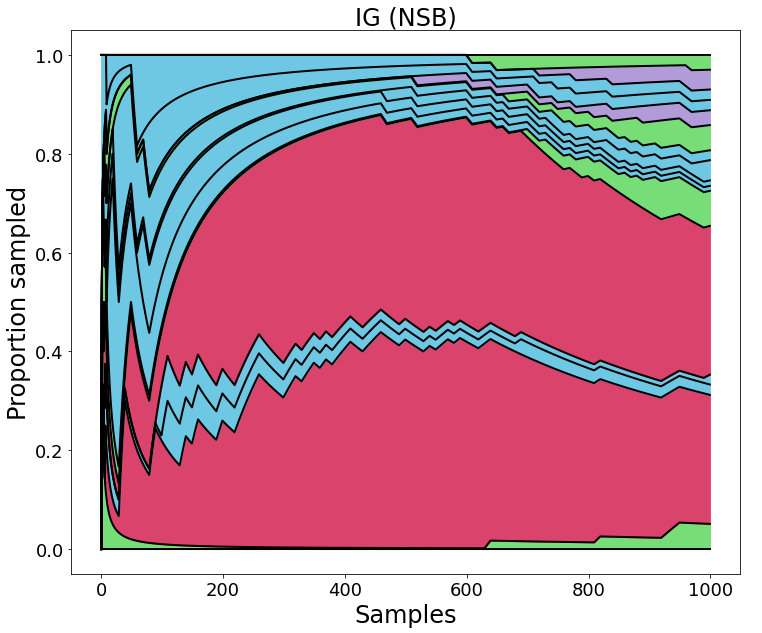}
    \includegraphics[width=0.21\textwidth]{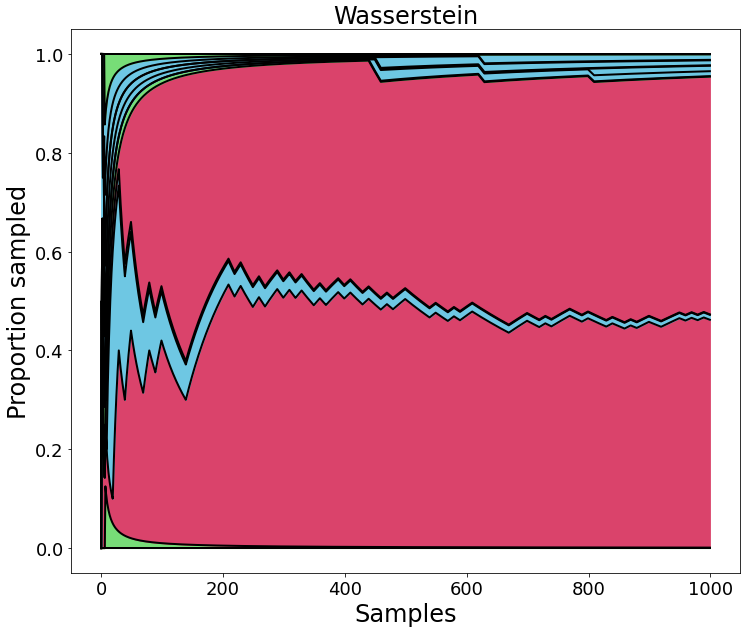}
    \includegraphics[width=0.21\textwidth]{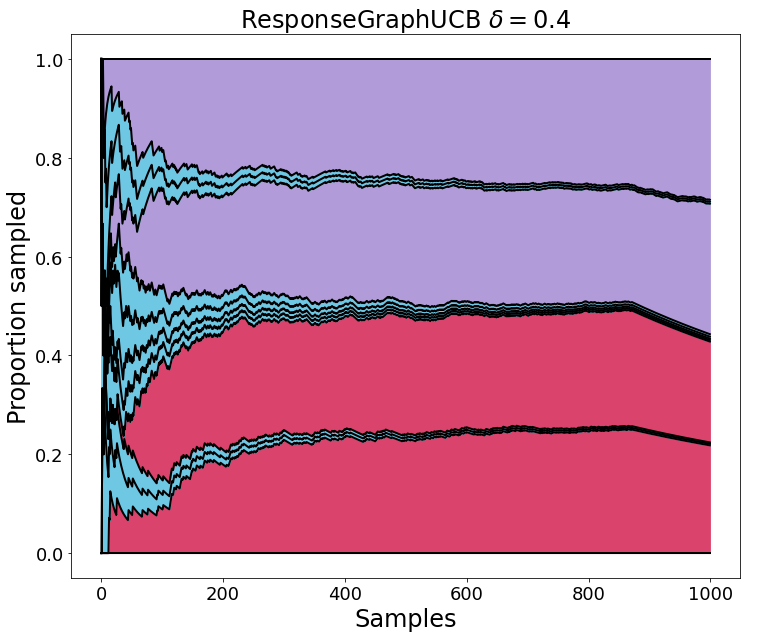}
    \caption{Proportion of entries sampled on 2 Good, 2 Bad for different methods and objectives.}
    \label{fig:2good2bad_entries}
\end{figure*}

\paragraph{Focusing on Relevant Payoffs} Figure \ref{fig:2good2bad_entries} presents the behavior of our method and RG-UCB on this task.
As expected, RG-UCB splits its sampling between the \textcolor{myred}{Red} entries and the \textcolor{mypurple}{Purple} entries, whereas our method concentrates its sampling much more significantly on the relevant entries, determining the ordering between agents 1 and 2.
This is because, in contrast to our method, RG-UCB aims to correctly determine the ordering between \textbf{all} entries used in the calculating of $\alpha$-rank, irrespective of whether they matter for the final outcome.

\paragraph{Wasserstein Payoff Selection Does Well} Comparing the Wasserstein Criterion with 
Information Gain payoff section, we can see that it enjoys better concentration of the sampling on the \textcolor{myred}{Red} entries, and improved performance towards the end of training. 
Appendix \ref{sec:2good_2bad_appendix} provides a more detailed analysis of this.

\begin{figure*}
    \centering
    \includegraphics[height=3.2cm]{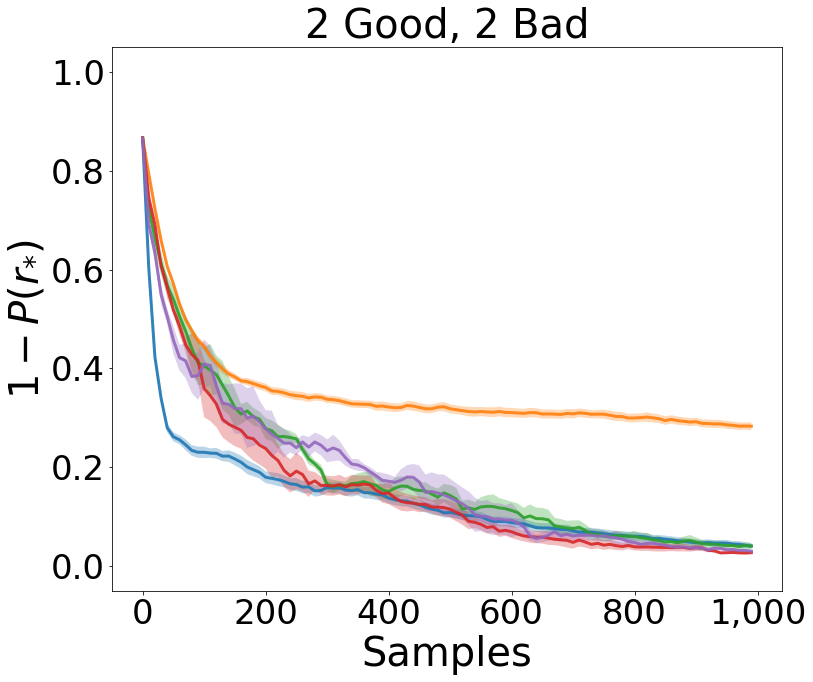}
    \includegraphics[height=3.2cm]{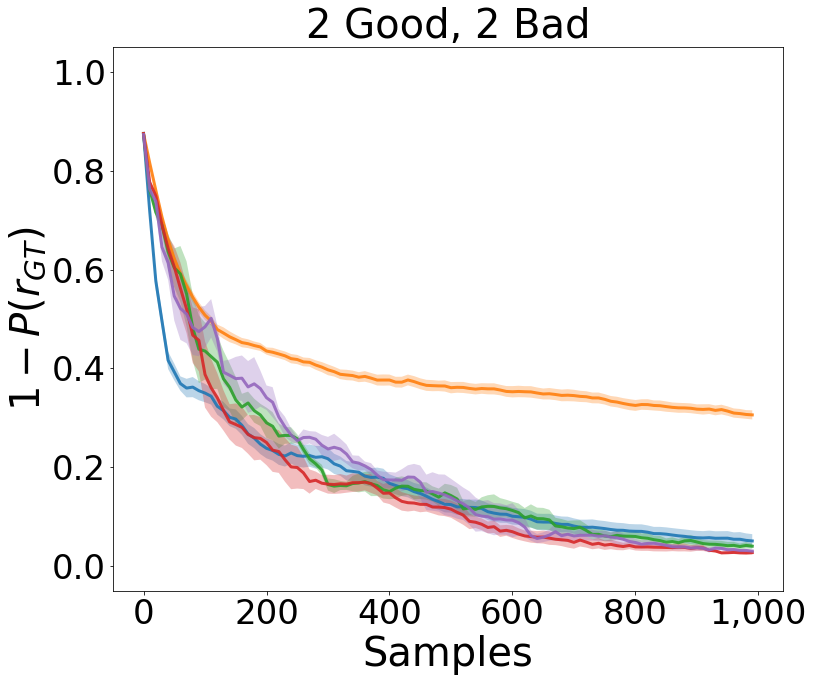}
    \includegraphics[height=3.2cm]{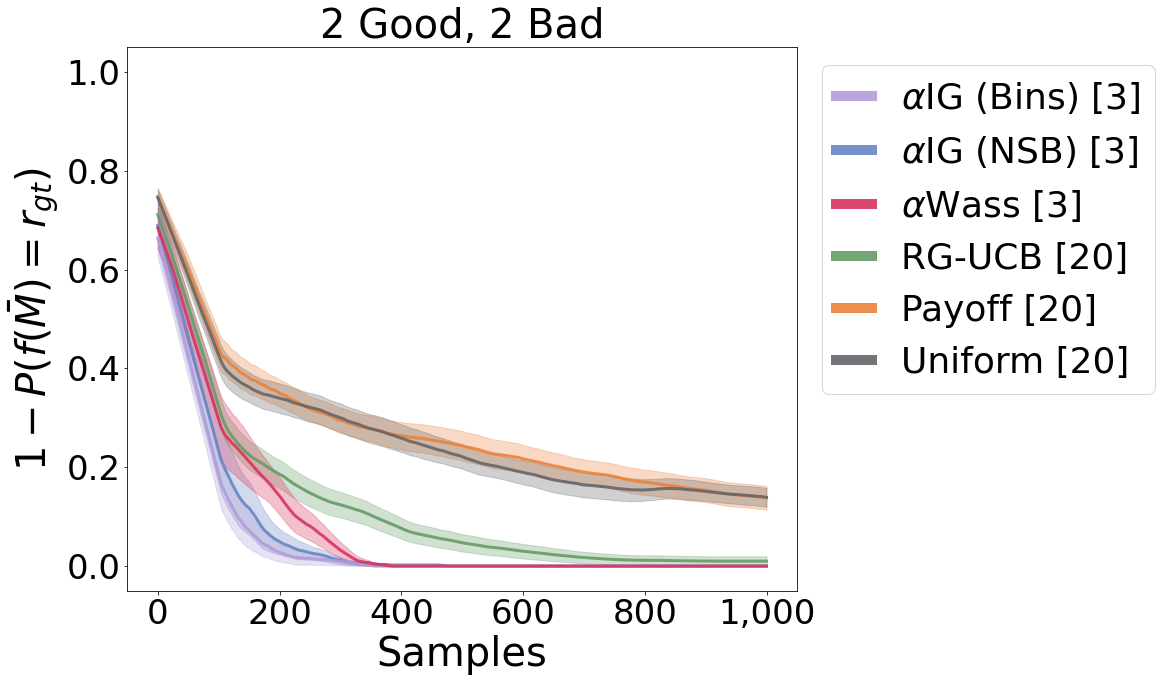}
    \caption{Results for 2 Good, 2 Bad. Graphs show the mean and standard error of the mean over multiple runs (shown in brackets) of 10 repeats each.}
    \label{fig:2good2bad_graphs}
\end{figure*}

\paragraph{Bayesian and Frequentist Regret Go Down} Figure \ref{fig:2good2bad_graphs} shows the resulting performance of the methods on this task, measured by the regret.
Due to the relative simplicity of the game, there is limited benefit to our method over RG-UCB, but there is a clear benefit over more naive methods that systematically or uniformly sample the entries.
We can see that the Bayesian regret $J^B_t$ and Frequentist regrets $J^F_t$ and $J^M_t$ are highly correlated, providing empirical justification for minimizing $J^B_t$ and validating that our method is concentrating on the ground truth.

\begin{figure*}
    \centering
    \includegraphics[height=3.2cm]{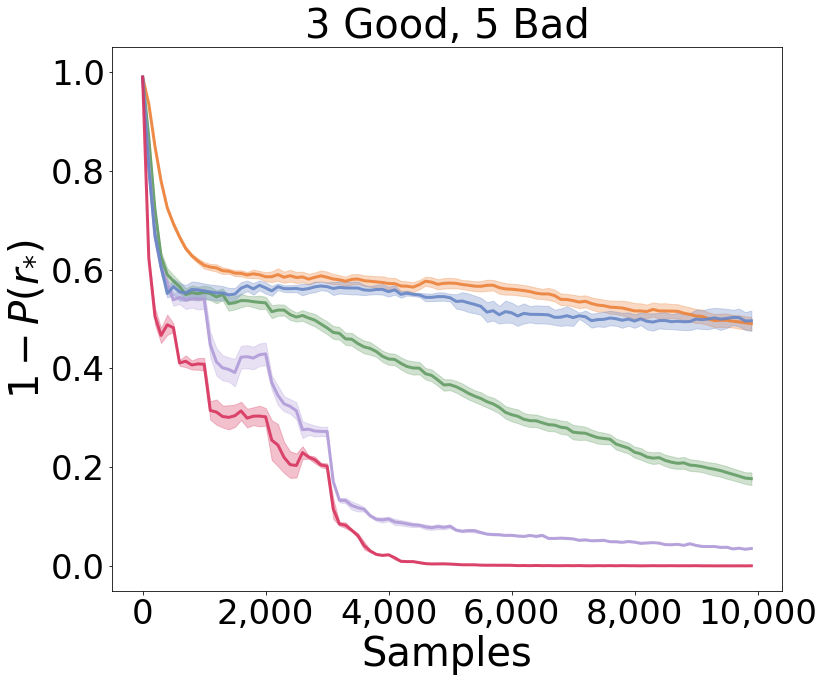}
    \includegraphics[height=3.2cm]{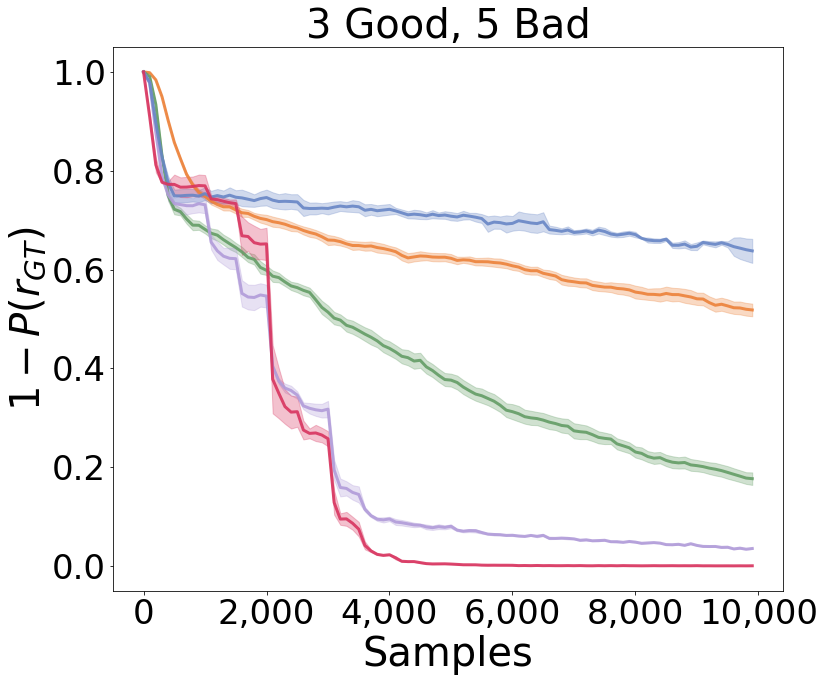}
    \includegraphics[height=3.2cm]{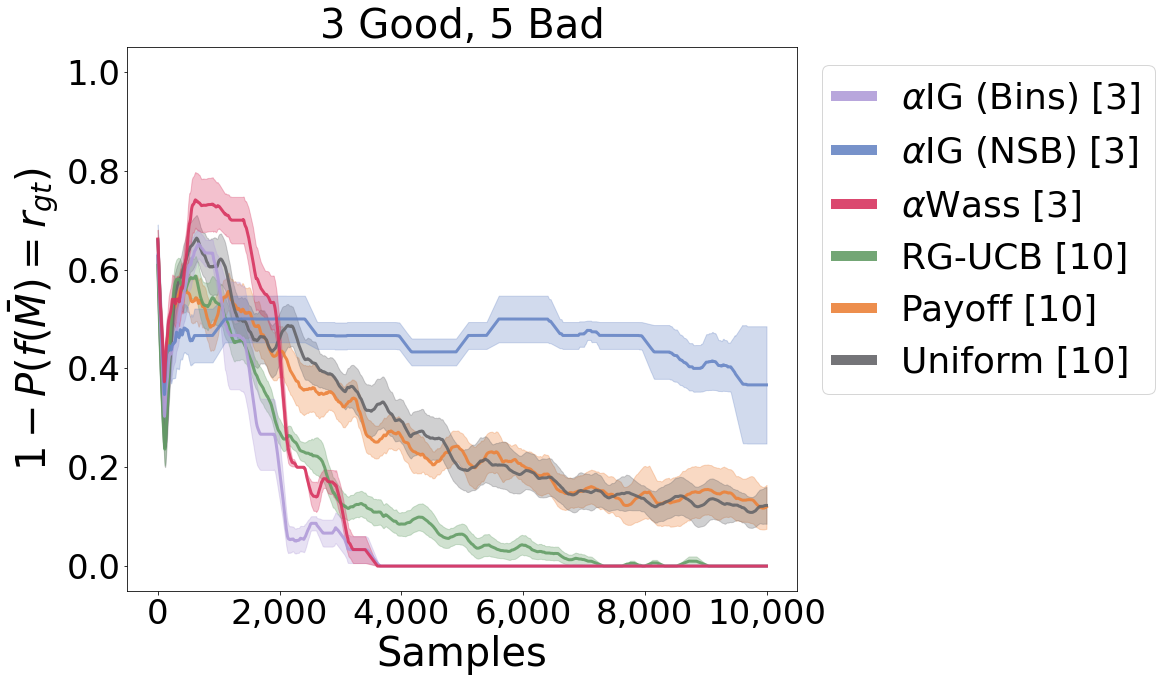}
    \caption{Results for 3 Good, 3 Bad. Graphs show the mean and standard error of the mean over multiple runs (shown in brackets) of 10 repeats each.}
    \label{fig:3good5bad_graphs}
\end{figure*}

\paragraph{Comparing Entropy Estimators}  We also investigate a larger scale version of 2 Good, 2 Bad with 3 good and 5 bad agents. 
Figure \ref{fig:3good5bad_graphs} shows the results, demonstrating a clear benefit for our method using the Binning estimator for the Information Gain or the Wasserstein objective.
The performance of the NSB entropy estimator is not surprising given the significantly larger nature of this task compared to `2 Good, 2 Bad'. 
A necessary part of the NSB estimator is an upper-bound on the total number of atoms in the distributions, for which we only have a crude approximation that grows exponentially with the size of the payoff matrix. 
\begin{figure*}
    \centering
    \includegraphics[width=0.21\textwidth]{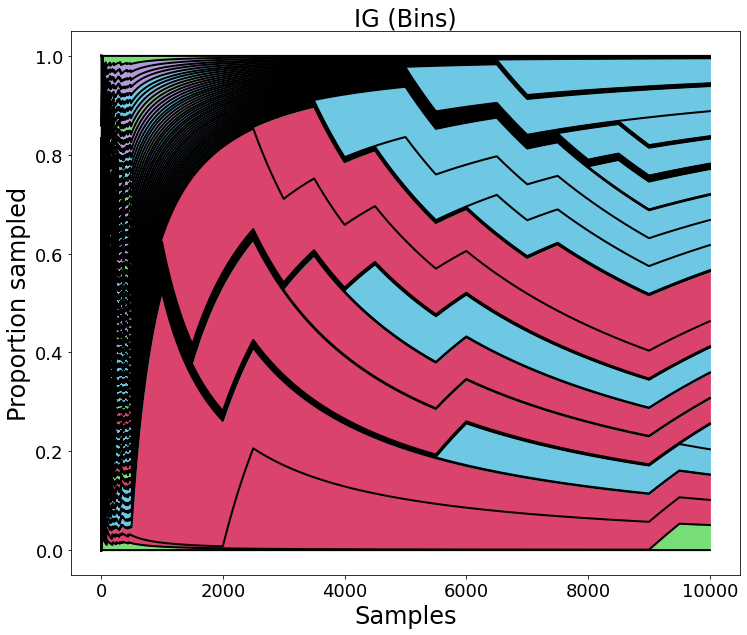}
    \includegraphics[width=0.21\textwidth]{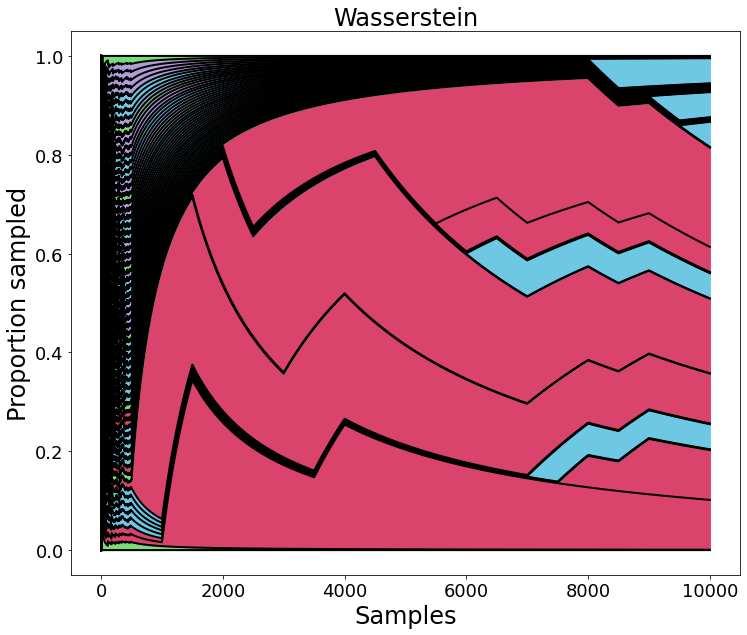}
    \includegraphics[width=0.21\textwidth]{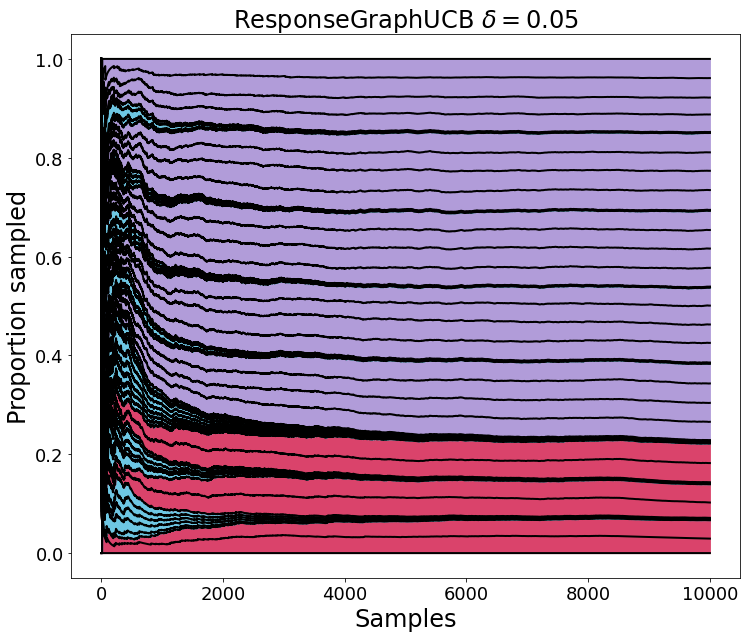}
    \caption{Proportion of entries sampled on 3 Good, 5 Bad.}
    \label{fig:3good5bad_entries}
\end{figure*}
Figure \ref{fig:3good5bad_entries} shows the proportion of entries sampled for $\alpha$IG (Bins), the Wasserstein objective, and RG-UCB.
Once again, RG-UCB spends a significant part of its sampling budget determining the ordering between agents that do not have an effect on the $\alpha$-rank of the game (in this task agents 3 to 8).
In contrast, our methods concentrate their sampling on the \textcolor{myred}{Red} entries that determine the payoffs between the top 3 agents, and hence the true $\alpha$-rank. In general, our algorithm does not depend as much on accurate estimates on entropy but on identifying the distribution with a lowest entropy, for which the NSB estimator isn't tuned.

\paragraph{Incorporating Prior Knowledge}

\begin{figure*}
    \centering
    \includegraphics[height=3.2cm]{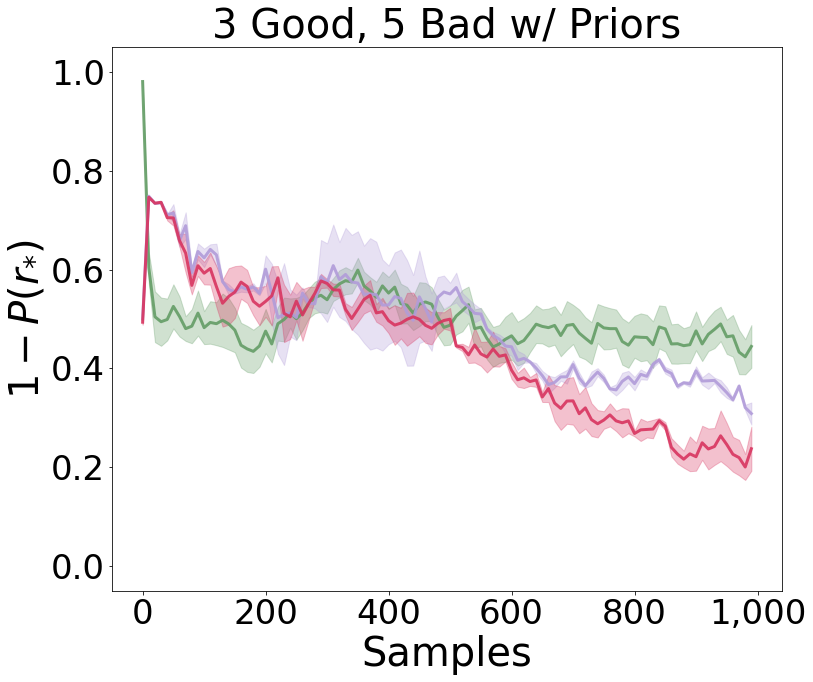}
    \includegraphics[height=3.2cm]{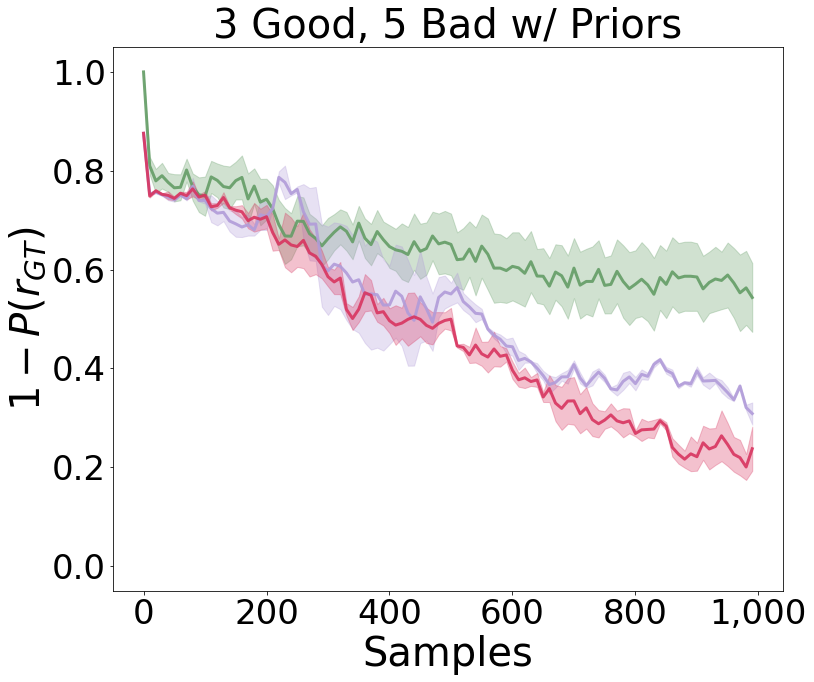}
    \includegraphics[height=3.2cm]{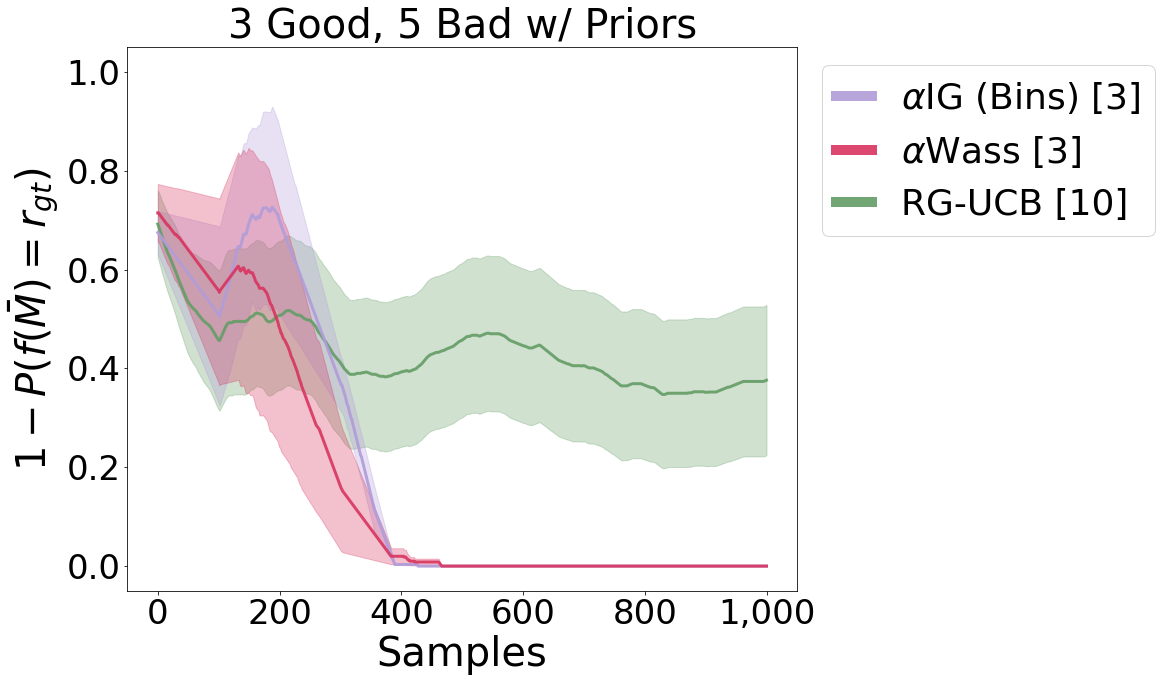}
    \caption{Results on 3 Good, 5 Bad when incorporating prior knowledge into the models.}
    \label{fig:3good5bad_kernel}
\end{figure*}

A large benefit of our Bayesian-based approach is the ability to incorporate prior knowledge and modelling assumptions into our model in a principled manner. 
To demonstrate the benefits, we incorporate the following prior knowledge into both our algorithm and RG-UCB:
1) $M(\sigma, \tau) + M(\tau, \sigma) = 1$. 2) Entries in their respective blocks are equal to each other (except for the top left block).
A detailed description of the setup is included in Appendix \ref{sec:3good5bad_kernel_setup}.
Figure \ref{fig:3good5bad_kernel} compares the performance of $\alpha$IG, $\alpha$Wass, and RG-UCB on 3 Good, 5 Bad when utilizing this prior knowledge.
We can see that our approach significantly outperforms RG-UCB on this task, further demonstrating the importance of our direct information gain objective. The results also show significantly improved sample efficiency over the results in Figure \ref{fig:3good5bad_graphs}, demonstrating that our $\alpha$IG and $\alpha$Wass are able to efficiently take advantage of the prior knowledge supplied.

\section{Conclusions}
We described $\alpha$IG, an algorithm for estimating the $\alpha$-rank of a game using a small number of payoff evaluations. $\alpha$IG works by maximizing information gain. It achieves competitive sample efficiency and allows a way of building in prior knowledge about the payoffs.

\section*{Acknowledgements}
We thank the Game Intelligence group at Microsoft Research Cambridge for their useful feedback, support, and help with setting up computing infrastructure. 
Tabish Rashid is supported by an EPSRC grant (EP/M508111/1, EP/N509711/1).

\bibliography{references}

\newpage
\onecolumn
\appendix

\section{Additional $\alpha$-Rank Background}

In this section we include a worked example on how to calculate the $\alpha$-Rank in the single-population setting in the infinite-$\alpha$ regime.
The purpose of this example is to help build intuition about $\alpha$-Rank in our particular setting. 
We will use 2 Good, 2 Bad matrix game, shown in Figure \ref{fig:2good_2bad_labelled} as an example, where the payoffs represent the expected win rate between agents.

\begin{figure}[h!]
    \centering
    \includegraphics[height=3.4cm]{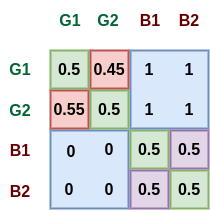}
    \caption{Payoff matrices for 2 Good, 2 Bad with the strategies for each player labeled.}
    \label{fig:2good_2bad_labelled}
\end{figure}

Note that in this example, we are considering a player playing the game against an identical opponent.
The strategy space is $S=\{G1, G2, B1, B2\}$, with $G1$ and $G2$ representing the \textit{good} agents, and $B1$ and $B2$ the \textit{bad} agents.
The good agents \textbf{always} win against the bad agents (hence $M(G1, B1) = 1$ and $M(B1, G1) = 0$ to use $G1$ and $B1$ as examples).
The $\alpha$-Rank is then a probability vector $r \in \Delta^{4-1} \subset \mathcal{R}^4$.

In order to compute the $\alpha$-Rank, we must first construct the Markov Chain whose nodes are elements of $S$.

Figure \ref{fig:2good_2bad_graph} is a diagram representing this Markov Chain.

\begin{figure}[h!]
    \centering
    \includegraphics[height=4.4cm]{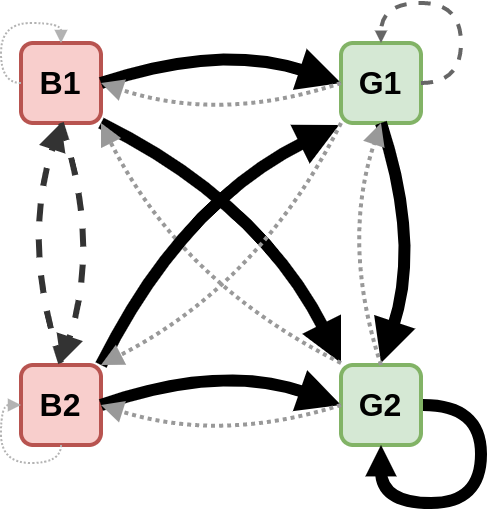}
    \caption{Markov Chain constructed to compute the $\alpha$-Rank. The transition probabilities between the nodes are represented using the width of the arrow.}
    \label{fig:2good_2bad_graph}
\end{figure}

The 4 nodes are the the elements of $S=\{G1, G2, B1, B2\}$, and the transition probabilities are as defined in Section \ref{sec:background}.
The width of the edges between the nodes in Figure \ref{fig:2good_2bad_graph} represent the magnitude of their transition probabilities.
Since the \textit{good} agents consistently beat the \textit{bad} agents there is a large transition probability from nodes $B1$ and $B2$ to nodes $G1$ and $G2$.
Likewise there is a very small transition probability from nodes $G1,G2$ to $B1, B2$.
Since strategies $B1$ and $B2$ are equally matched, there is an equal probability of transition from $B1$ to $B2$ as there is from transitioning from $B2$ to $B1$. 
Crucially though, that transition probability is significantly smaller than the transition probabilities from $B1$ and $B2$ to nodes $G1$ and $G2$.
All nodes' self-transition probabilities ensure that there is a valid transition matrix for the Markov Chain.
Importantly, $B1$ and $B2$ have small self-transition probabilities whereas $G2$ has a very large self-transition probability.

We know $(|S| - 1)^{-1} = \frac{1}{3}$, and letting $\epsilon>0$ be small, we can compute the exact transition matrix:

\[
\renewcommand\arraystretch{2}
\bordermatrix{ 
   & G1 & G2 & B1 & B2  \cr
G1 & (2-\epsilon)/3  &  (1-\epsilon)/3 & \epsilon/3 & \epsilon/3  \cr
G2 & (1-\epsilon)/3  &  1-\epsilon      & \epsilon/3 & \epsilon/3  \cr
B1 & (1-\epsilon)/3  & (1-\epsilon)/3 & (1+4\epsilon)/6 & 1/6  \cr
B2 & (1-\epsilon)/3  & (1-\epsilon)/3 & 1/6 & (1+4\epsilon)/6 
} \qquad
\]

Intuitively, we can see that the transition probabilities are all \textit{leading} to node $G2$. 
So we would expect to spend a large proportion of time at node $G2$, if we were \textit{traversing} the graph according to the transition probabilities. 
Hence, $G2$ is a \textit{stronger} strategy than the others and so we would expect the $\alpha$-Rank to reflect that. 

This is formalised in the computation of $\alpha$-Rank by considering the unique stationary distribution of the Markov Chain.
The $\alpha$-Rank is then $(0,1,0,0)$ ($G1=0,G2=1,B1=0,B2=0$) as $\epsilon \to 0$.
\section{Implementation Details}

\subsection{$\alpha$IG (Bins).}
For this binning entropy estimator we split $[0,1]$ into 101 equal bins of width $0.005$ (implemented by rounding to the nearest second decimal place).
We then estimated the entropy using a histogram.

\subsection{$\alpha$IG (NSB).}
The NSB estimator requires an upper bound on the total number of atoms, but since we do not know the true upper bound we utilize an estimate on the total number of possible $\alpha$-ranks, which we describe below. 
We use the open-source implementation provided in \citep{2020ndd}.

\subsection{Upper bound on number of $\alpha$-ranks}

In the infinite-$\alpha$ regime there are a finite number of possible $\alpha$-ranks.
This is because only the ordering between relevant entries in the payoff matrix changes the transition matrix of the Markov Chain produced in the computation of $\alpha$-rank \citep{rowland2019multiagent}.

Let there be $k$ populations each with $S$ strategies.
Then there are $S^k$ strategies considered and so the transition matrix of the Markov Chain has $S^k$ rows, one for each of the possible joint-strategies.

Each possible joint-strategy $\sigma$ can transition to at most $k(S-1)$ other strategies $\tau \neq \sigma$. The probability of a self-transition is uniquely determined based on these probabilities. 

This gives at most $2^{(k(S-1))}$ unique values for that row. 

There are then $[2^{(k(S-1))}]^{S^k} = 2^{S^k(k(S-1))}$ unique transition matrices.
Thus, the possible number of unique $\alpha$-ranks is upper-bounded by $2^{S^k(k(S-1))}$.
This bound is not tight, since there are many transition matrices with equal stationary distributions.

In our experiments with $K=1$ this gives $2^{S(S-1)}$.

\subsection{Conditioning of the belief distribution}
\label{sec:cond_payoff}
In our experiments we found that setting $N_{c}=1$ as suggested by theory is not always sufficient and use $N_{c}=100$ for all experiments.
\begin{figure}[h]
    \centering
    \includegraphics[width=0.32\textwidth]{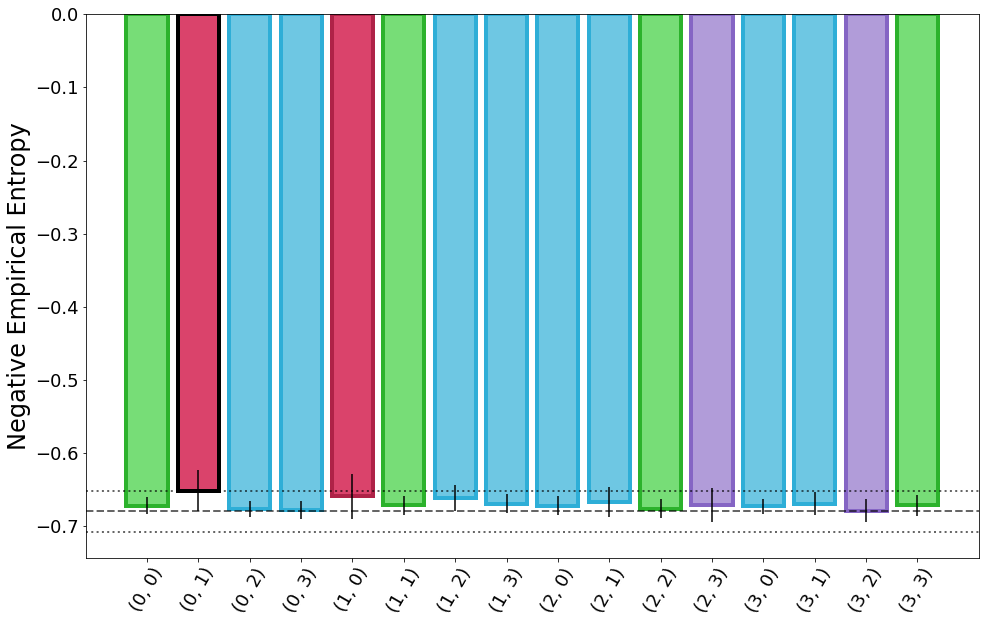}
    \includegraphics[width=0.32\textwidth]{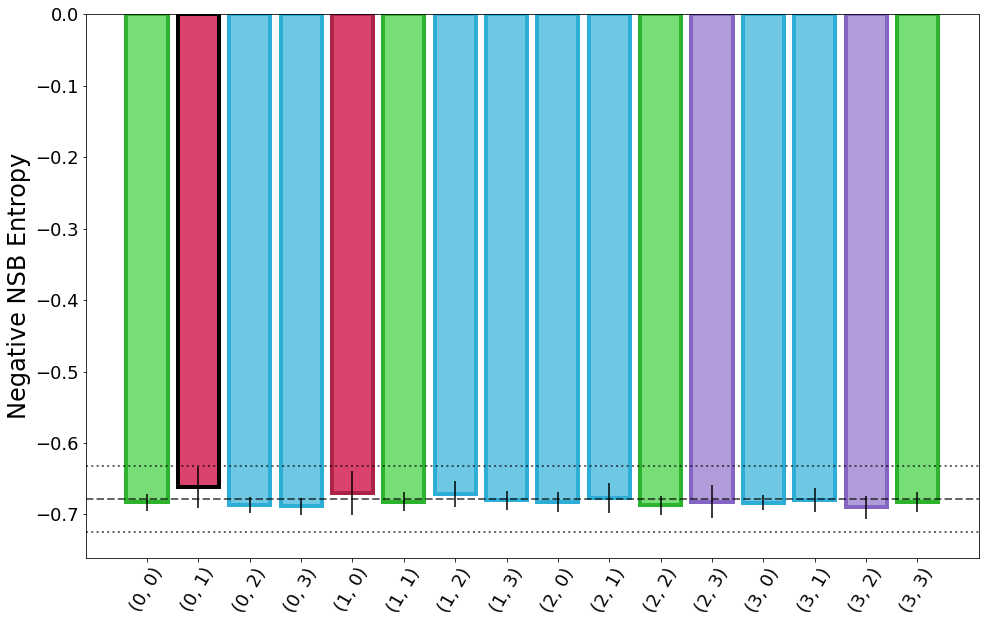}
    \includegraphics[width=0.32\textwidth]{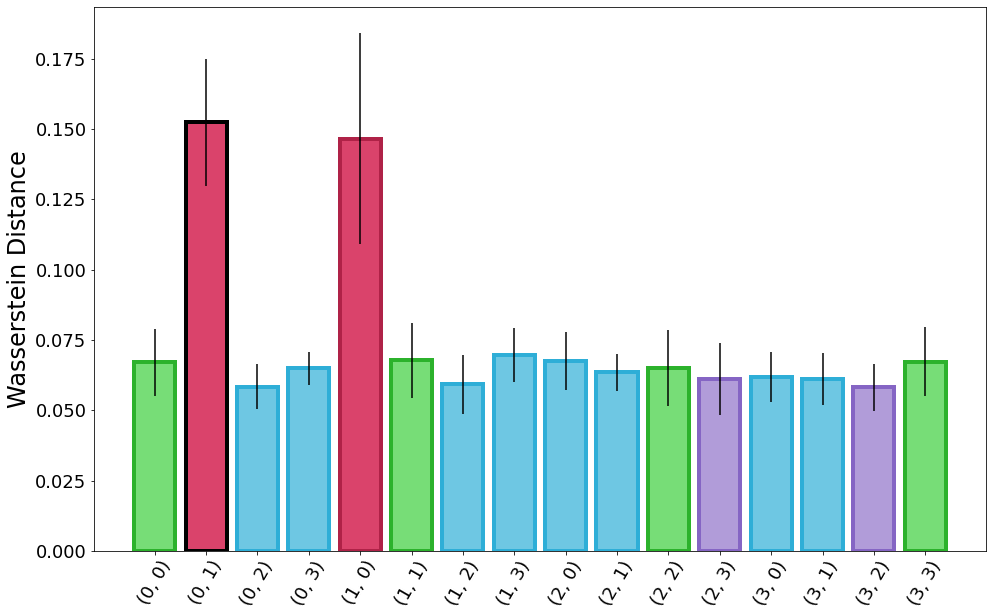}
    \includegraphics[width=0.32\textwidth]{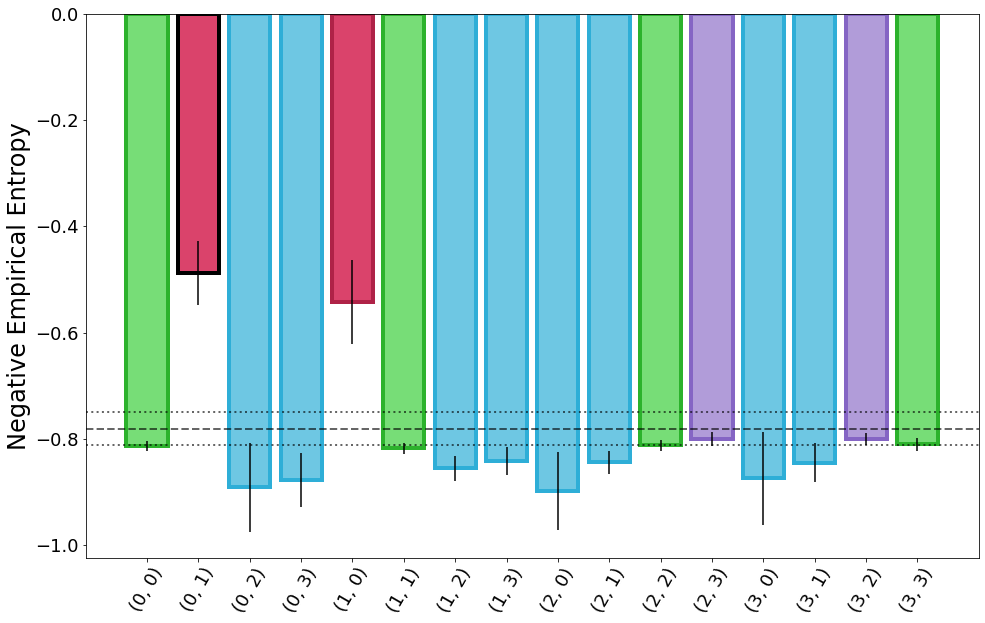}
    \includegraphics[width=0.32\textwidth]{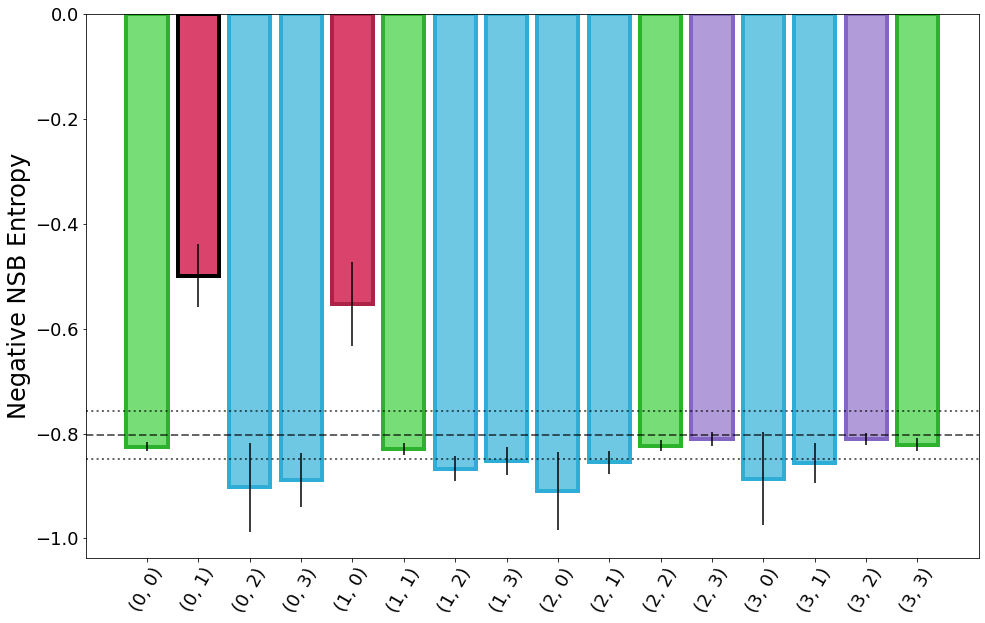}
    \includegraphics[width=0.32\textwidth]{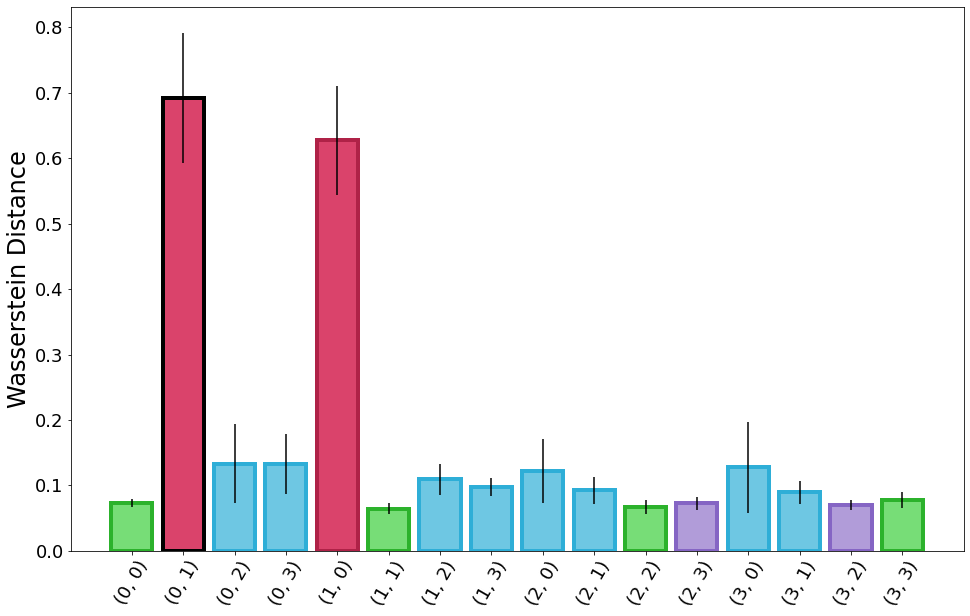}
    \caption{Comparing the values of the objectives for each entry after sampling 5 values for every entry.
    \textbf{Top} shows the results for $N_{c}=1$.
    \textbf{Bottom} shows the results for $N_{c}=100$.
    Mean and standard deviation are plotted across 10 seeds, maximum entry is highlighted in black. The mean and standard deviation of the (estimated) entropy of the current belief distribution is also plotted as a dashed horizontal line.
    }
    \label{fig:1cond_vs100cond_5_entries}
\end{figure}

\begin{figure}
    \centering
    \includegraphics[width=0.32\textwidth]{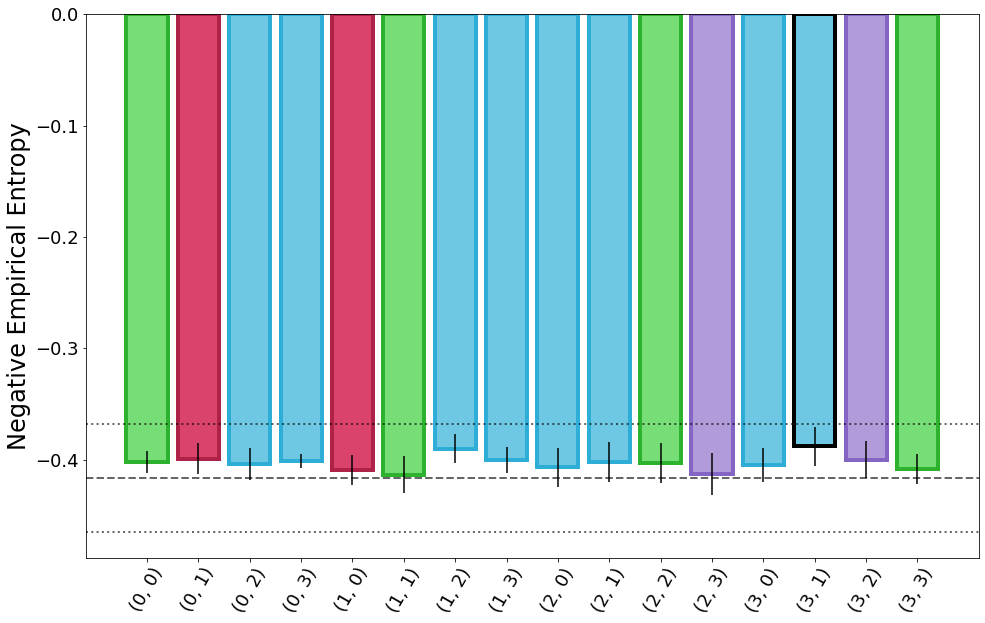}
    \includegraphics[width=0.32\textwidth]{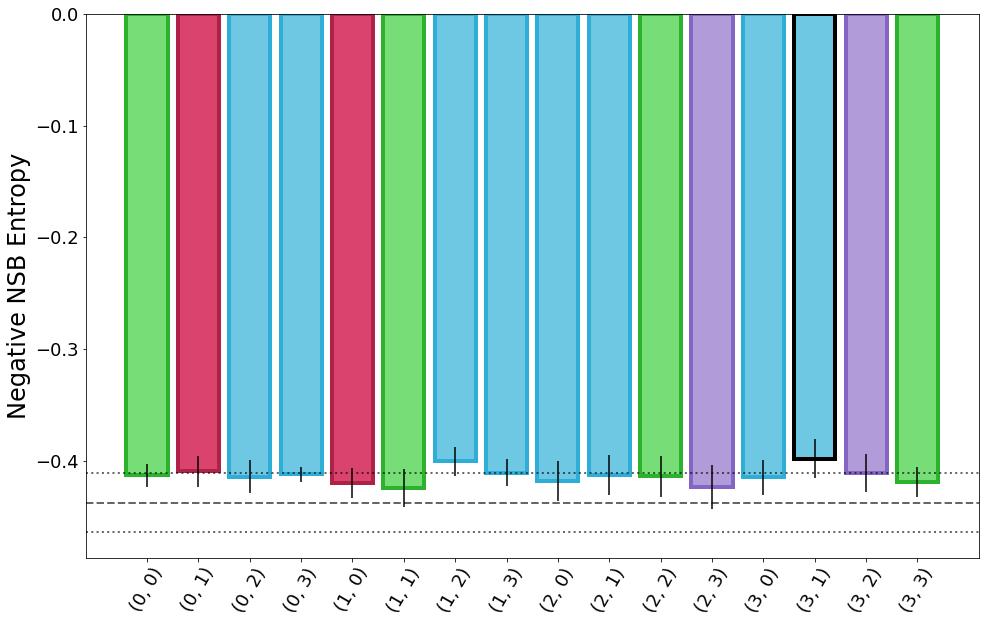}
    \includegraphics[width=0.32\textwidth]{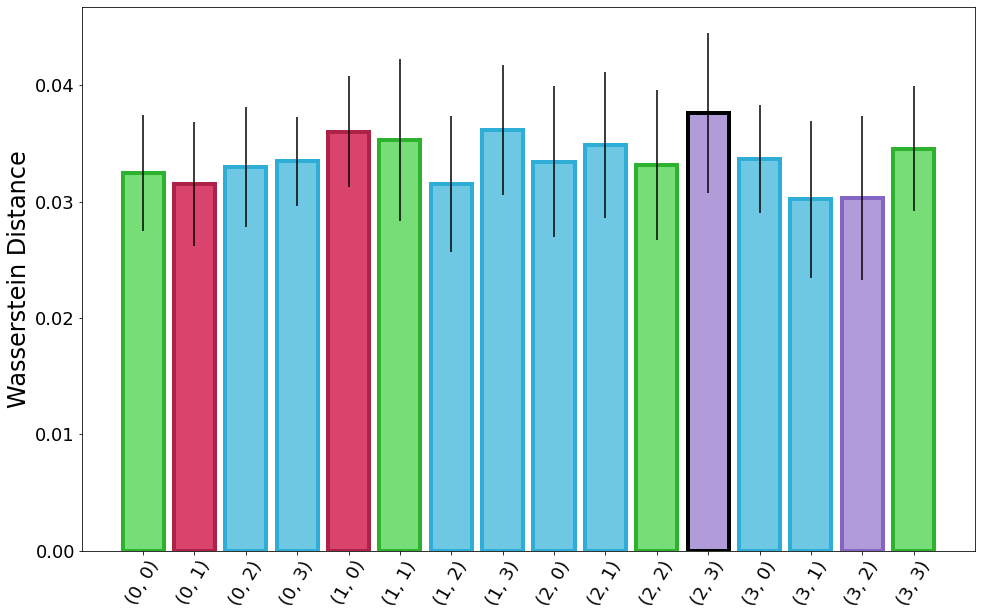}
    \includegraphics[width=0.32\textwidth]{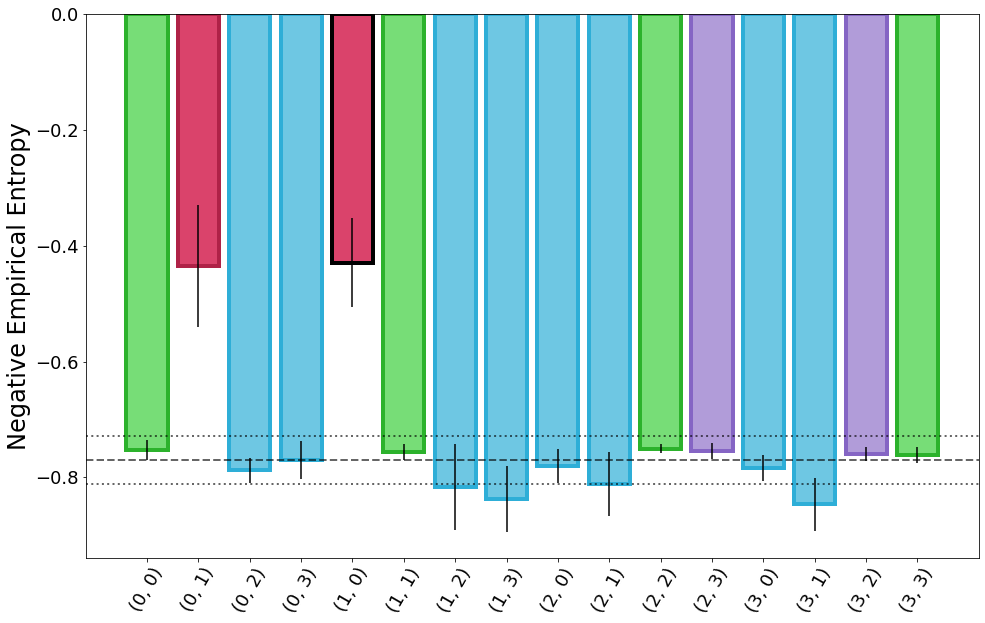}
    \includegraphics[width=0.32\textwidth]{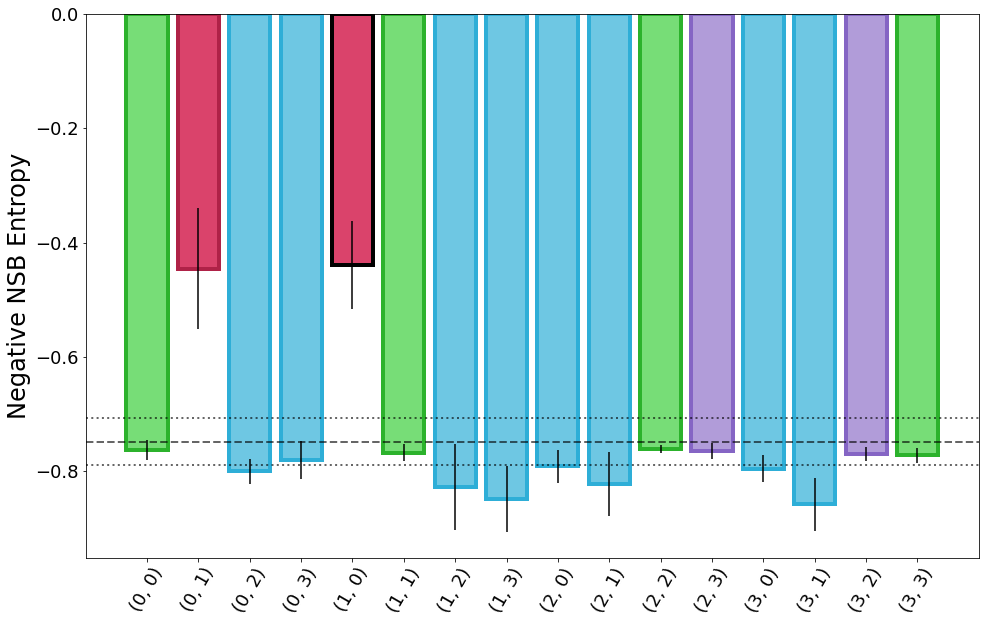}
    \includegraphics[width=0.32\textwidth]{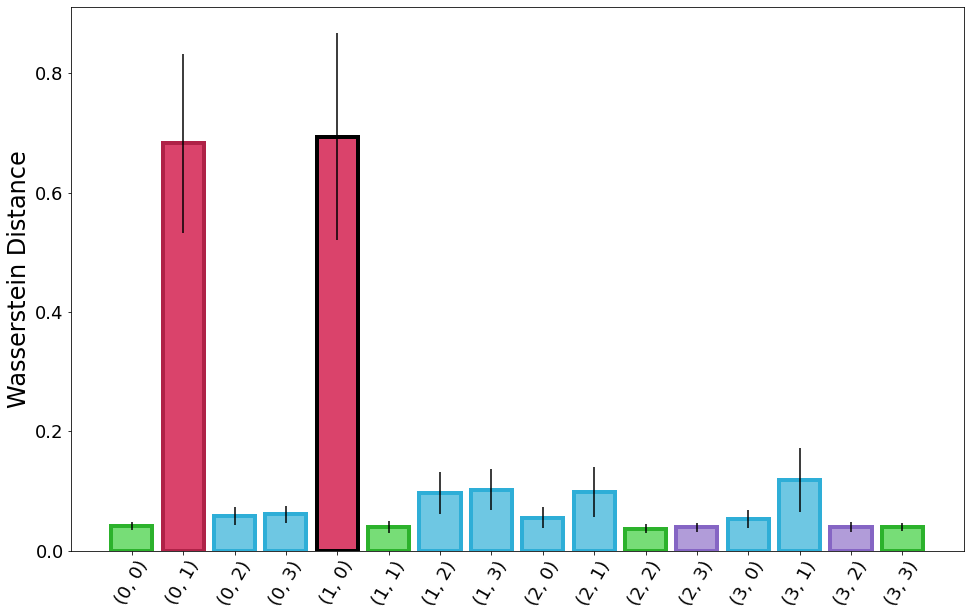}
    \caption{Comparing the values of the objectives for each entry after sampling 5 values for every entry, and then additionally sampling 250 values for the Red entries.
    \textbf{Top} shows the results for $N_{c}=1$.
    \textbf{Bottom} shows the results for $N_{c}=100$.
    Mean and standard deviation are plotted across 10 seeds, maximum entry is highlighted in black. The mean and standard deviation of the (estimated) entropy of the current belief distribution is also plotted as a dashed horizontal line.
    }
    \label{fig:1cond_vs100cond_5_255_entries}
\end{figure}

After drawing a sample $m'_t \sim M_t + \epsilon$, we then condition our belief distribution over $\alpha$-rank on this sample $N_{c}$ times and then approximate the Entropy of the resulting hallucinated belief distribution (or the Wasserstein distance between the current belief and the hallucinated belief).
Theory suggests that setting $N_{c}=1$ is sufficient, however empirically we found that this did not produce satisfactory results.
Figure \ref{fig:1cond_vs100cond_5_entries} shows that only conditioning once produces very little separation between the values for the different entries.
Additionally, we can see that there is very little separation between the current belief's entropy and the hallucinated belief's entropy.
In contrast, we can see that conditioning $100$ times produces significantly more separation.
Figure \ref{fig:1cond_vs100cond_5_255_entries} shows the same trend, after additionally sampling 250 values for the red entries.
The Wasserstein objective shows the same trend, that conditioning more than once produces significantly more separation. 
A Wasserstein Distance of 0 indicates that the two distributions are identical.
\section{Experimental Setup}
\label{sec:experimental_setup}

\subsection{$\alpha$-Rank}
In the computation of $\alpha$-rank we set $\epsilon = 10^{-6}$ in all of our experiments.
\subsection{Baselines}

\textbf{ResponseGraphUCB}, uses a Hoeffding Bound to construct the confidence interval: \\
$(\mu - \sqrt{\frac{\log(2/\delta)(b-a)}{2N}}, \mu + \sqrt{\frac{\log(2/\delta)(b-a)}{2N}})$.
Where $\delta$ is the confidence hyperparameter of the algorithm, $b$ is the maximum value an entry can take, $a$ is the minimum value, and $N$ is the number of times a value has been seen for an entry.
For all experiments we swept over $\delta \in \{0.4, 0.3,0.2,0.1,0.05,0.01,0.001\}$, and the final value is selected by considering the area under the curve for $1 - P(f(\bar{M}) = r_{GT})$.

\textbf{Uniform}. The entry to sample if picked uniformly from all possible entries.

\textbf{Payoff}. The entry which maximises the information gain between its sample and the payoff distribution is chosen. For an isotropic Gaussian this is equivalent to picking the entry which the lowest count, which results in systematic sampling of each entry.
For a non isotropic Gaussian the same procedure as \citep{srinivas2009gaussian} is used.

\subsection{Graphs}

\textbf{$1 - P(f(\bar{M}) = r_{GT})$}. At each timestep we compute the $\alpha$-rank of the mean payoff matrix.
Equality is determined if $|f(\bar{M}) = r_{GT})|_1 < 0.01$.
The choice of $0.01$ is largely arbitrary, we did not find the results to be sensitive to this. 

\textbf{$1 - P(r_{GT})$}. 100 times during training (evenly spaced), we sample $2000$ samples from the current belief distribution over $\alpha$-ranks. $P(r_{GT})$ is determined from these $2000$ samples (which are aggregated by rounding each value to the nearest $3 d.p.$) by counting the number of sampled $\alpha$-ranks $r$ such that $|r - r_{GT}|<0.01$.

\textbf{$1 - P(r_{*})$}, is determined similarly to $1 - P(r_{GT})$, except we use the $2000$ samples to calculate the mode.

For ResponseGraphUCB, we construct a distribution over the payoff entries as being uniform over the confidence intervals. 

\subsection{Environments}
\subsection{2 Good, 2 Bad}

Observations are sampled from $Ber(x)$, where $x$ is the value in the payoff matrix.

\textbf{$\alpha$IG(Bins), $\alpha$IG(NSB), $\alpha$Wass}.
Prior used is $\mathcal{N}(\mu_0, \sigma_0^2)$, with aleatoric noise $\sigma_A^2$.
$\mu_0 = 0.5$.
Swept over $\sigma_0^2 \in \{0.5, 1\}, \sigma_A^2 \in \{0.25,0.5,1\}$.
20 samples are used to approximate the expectation, $N_{e}=20$.
1000 samples are drawn from the belief distribution(s) to approximate the quantities inside the expectation, $N_{b}=1000$.
We set $N_{r}=10$.

For all 3 methods we set $\sigma_0^2=1$.
For $\alpha$IG (Bins) and $\alpha$IG (NSB) we set $\sigma_A^2=0.5$, and for $\alpha$Wass we set $\sigma_A^2=0.25$.

\textbf{ResponseGraphUCB}. We set $\delta=0.4$. Maximum value is 1, minimum value is 0.

\subsection{3 Good, 5 Bad}

Observations are sampled from $Ber(x)$, where $x$ is the value in the payoff matrix.

\textbf{$\alpha$IG(Bins), $\alpha$IG(NSB), $\alpha$Wass}.
Prior used is $\mathcal{N}(\mu_0, \sigma_0^2)$, with aleatoric noise $\sigma_A^2$.
$\mu_0 = 0.5$.
Swept over $\sigma_0^2 \in \{0.5, 1\}, \sigma_A^2 \in \{0.25,0.5,1\}$.
$N_{e}=10$.
$N_{b}=500$.
$N_{r}=500$.

For all 3 methods we set $\sigma_0^2=1$.
For $\alpha$IG (Bins) and $\alpha$IG (NSB) we set $\sigma_A^2=0.5$, and for $\alpha$Wass we set $\sigma_A^2=0.25$.

\textbf{ResponseGraphUCB}. We set $\delta=0.05$.

\subsection{4x4 Gaussian}

To match the games considered in our theoretical analysis, Observations are sampled from $\mathcal{N}(x,1)$ and then clipped to be within $1$ of $x$, where $x$ is the value of the entry in the payoff matrix.
The values of $x$ are uniformly drawn from $[0,1)$.

\textbf{$\alpha$IG(Bins), $\alpha$IG(NSB), $\alpha$Wass}.
Prior used is $\mathcal{N}(\mu_0, \sigma_0^2)$, with aleatoric noise $\sigma_A^2$.
$\mu_0 = 0.5$.
Swept over $\sigma_0^2 \in \{0.5, 1\}, \sigma_A^2 \in \{0.5,1\}$.
$N_{e}=10$.
$N_{b}=500$.
$N_{r}=100$.
For $\alpha$IG(Bins) we set $\sigma_0^2 = 1$ and $\sigma_A^2=1$.
For $\alpha$IG(NSB) we set $\sigma_0^2 = 0.5$ and $\sigma_A^2=0.5$.
For $\alpha$Wass we set $\sigma_0^2 = 1$ and $\sigma_A^2=0.5$.

\textbf{ResponseGraphUCB}. We set $\delta=0.3$.  Maximum value is 2, minimum value is -1.

\subsection{3 Good, 5 Bad Incorporating Prior Knowledge}
\label{sec:3good5bad_kernel_setup}

\textbf{$\alpha$IG(Bins), $\alpha$Wass}.
Aleatoric noise $\sigma_A^2$ is set to $0.5$ and the mean $\mu = 0.5$.
$\mu_0 = 0.5$.
Swept over $\sigma_A^2 \in \{0.25,0.5,1\}$.
$N_{e}=10$.
$N_{b}=500$.
$N_{r}=100$.

The kernel $K$ we use for the GP is specified as follows:\\
In order to encode the prior knowledge that elements within a block (except the top left block) are equal, we partition the payoff matrix. A strategy $\sigma = (x,y)$ is a member of block $b_1$ if $1 \leq x \leq 3$ and $1 \leq y \leq 3$, if it is a payoff between a Good agent and another Good agent.
Block $b_2$ if $1 \leq x \leq 3$ and $y > 3$, a payoff between a Good agent and a Bad agent.
Block $b_3$ if $1 \leq y \leq 3$ and $x > 3$, a payoff between a Bad agent and a Good agent.
Block $b_4$ otherwise, if it is a payoff between a Bad agent and a Bad agent.

The kernel $k$ encoding block-wise equality is then defined as:\\
$k(b_i, b_j) = 1$ for all $i \neq j$.\\
$k(b_i, b_i) = 0$ for $i \in \{2,3,4\}$.\\
$k(b_1,b_1) = 1$.

To additionally encode anti-symmetry (about the mean $\mu=0.5$) we then produce a new kernel $k'$ from $k$ as follows:\\
For a strategy $\sigma = (x,y)$, define $\sigma^t := (y,x)$ the transpose of $\sigma$.
We wish to encode that $M(\sigma) = 1 - M(\sigma^t)$.\\
$k'(\sigma, \tau) = k(\sigma, \tau) + k(\sigma^t, \tau^t) - k(\sigma, \tau^t) - k(\sigma^t, \tau)$.

The finished kernel $K$ is then defined as $K := (k')^T(k')$ to ensure it is positive definite. The entries are then divided by 500 to ensure a suitable magnitude for the variance.

\textbf{ResponseGraphUCB}. We set $\delta=0.05$.

In order to incorporate the same modelling assumptions into RG-UCB, for every real sample we receive from the payoff matrix, we pretend to receive an appropriate sample for the relevant entries.

After receiving a payoff $p \in \{0,1\}$ for strategy $\sigma$, we then pretend to receive:\\
$1 - p$ for $\sigma^t$ to encode anti-symmetry.\\
$p$ for all $\tau$ in the same block as $\sigma$ (except the top-left block), where blocks are defined the same as for the kernel $k$ specified above.

\section{Further Results}

\subsection{Gaussian Games}

\begin{figure}
    \centering
    \includegraphics[height=3.25cm]{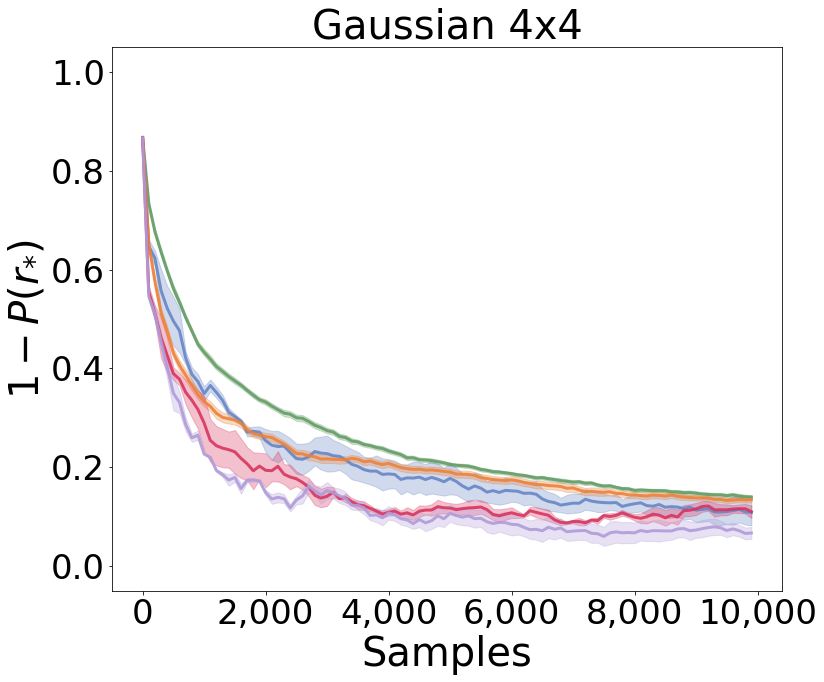}
    \includegraphics[height=3.25cm]{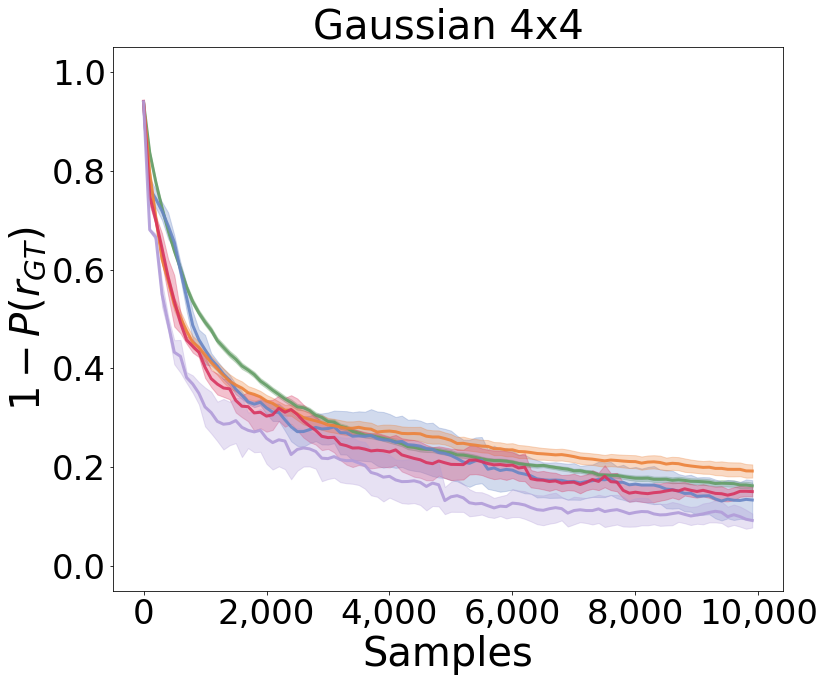}
    \includegraphics[height=3.25cm]{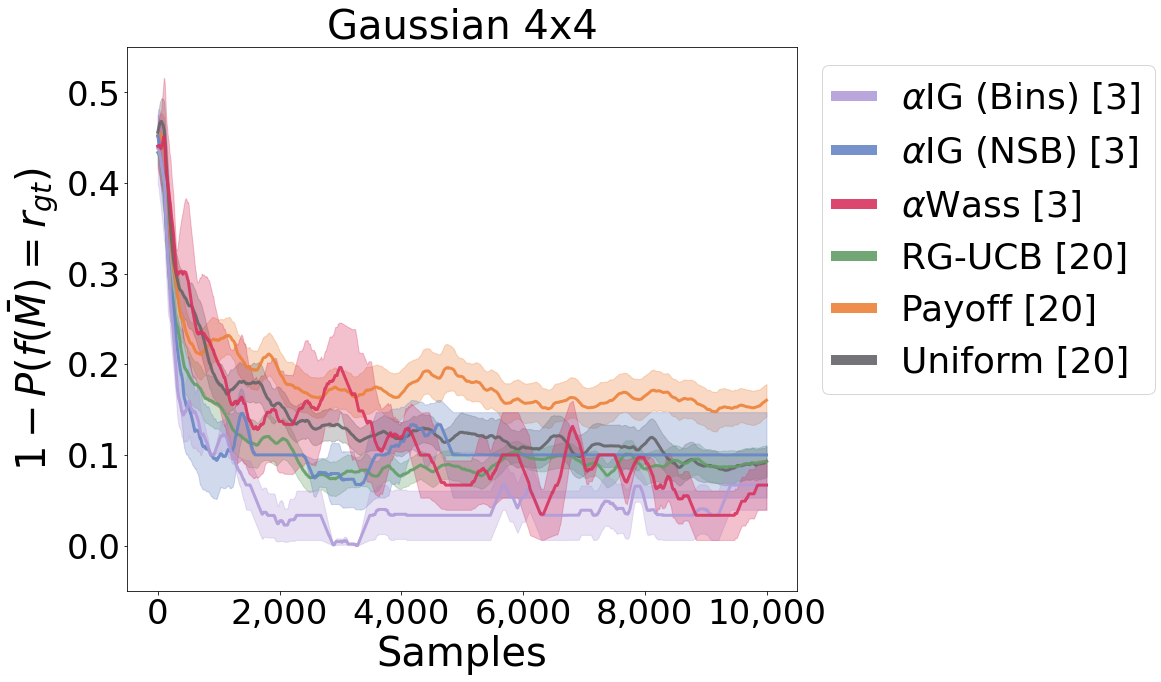}
    \caption{Results for the 4x4 Gaussian Game. Graphs show the the mean and standard error of the mean over multiple runs (shown in brackets) of 10 randomly sampled games.}
    \label{fig:4gauss_graphs}
\end{figure}

Figure \ref{fig:4gauss_graphs} shows the results on 4x4 games with Gaussian noise, demonstrating improved performance across all 3 regret metrics for the $\alpha$IG (Bins).
This is empirical confirmation of our theoretical results, and shows that our method achieves better performance compared to RG-UCB on general games. 

\subsection{2 Good, 2 Bad}
\label{sec:2good_2bad_appendix}
Figure \ref{fig:2good2bad_snapshot} shows the values used by the different objectives during training. The top row shows the values after sampling 5 values for each entry, showing a clear seperation between the \textcolor{myred}{Red} entries and the rest.
The bottom row shows the values after additionally sampling 250 values for the \textcolor{myred}{Red} entries.
We can then see a large difference between the Wasserstein and Entropy-based objectives.
As desired the Wasserstein-based objective shows a large separation between the \textcolor{myred}{Red} entries and the others, additionally assigning the smallest values to the irrelevant \textcolor{mygreen}{Green}, and \textcolor{mypurple}{Purple} entries.

\begin{figure}[h!]
    \centering
    \includegraphics[width=0.45\textwidth]{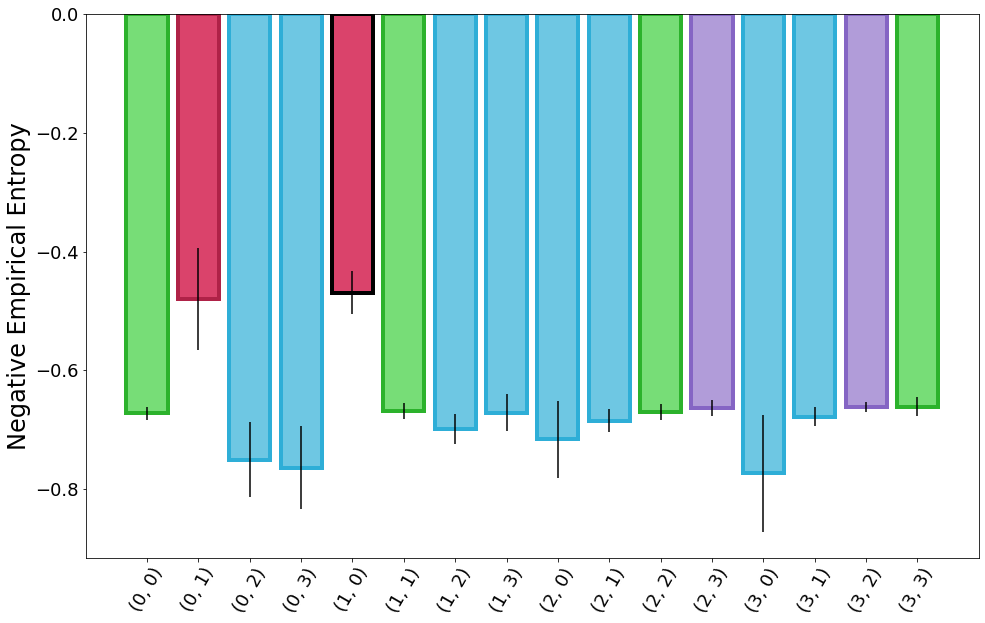}
    \includegraphics[width=0.45\textwidth]{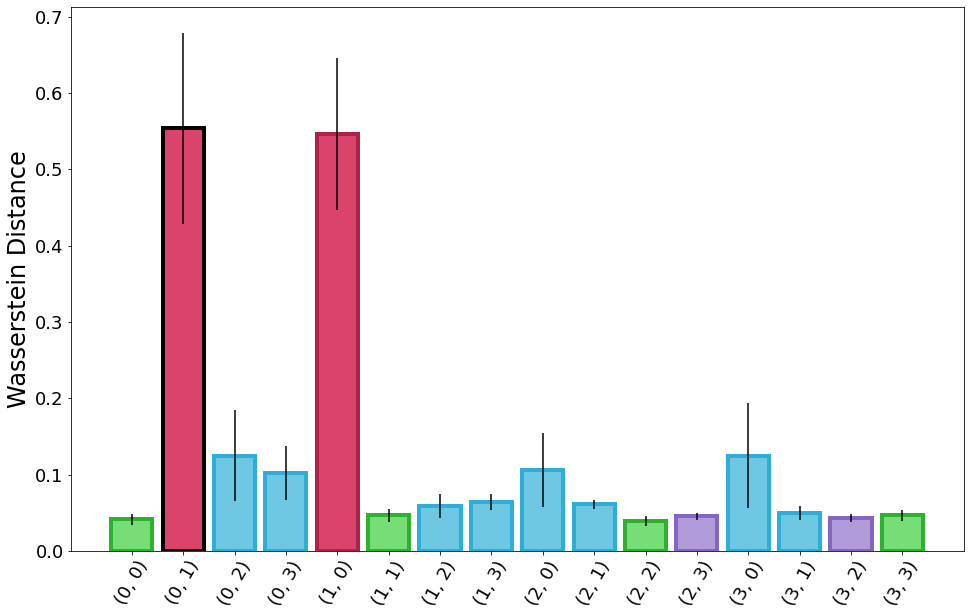}\\
    \includegraphics[width=0.45\textwidth]{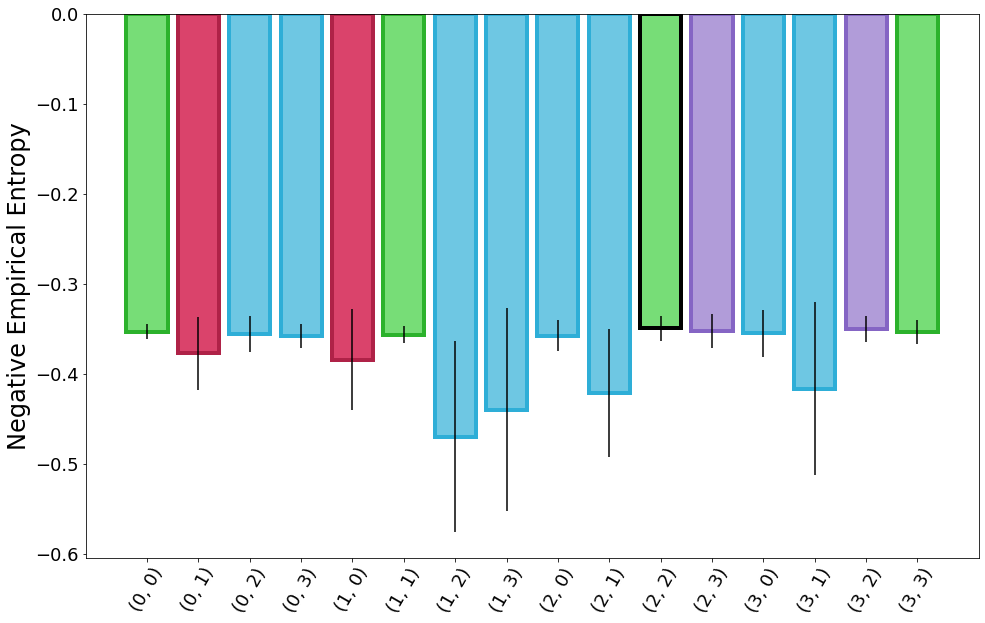}
    \includegraphics[width=0.45\textwidth]{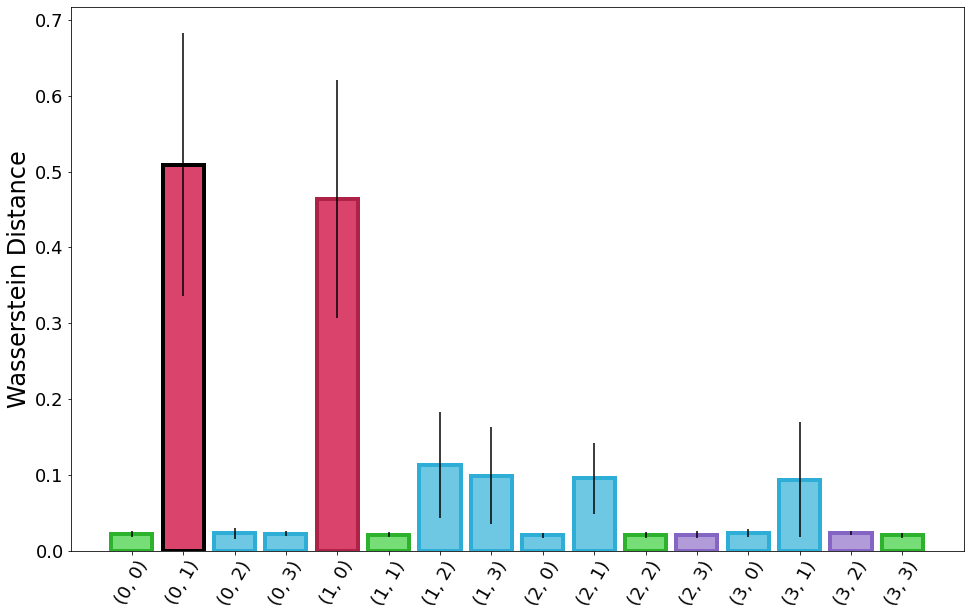}
    \caption{Value of the objectives for each entry after sampling 5 values for every entry (top) and additionally sampling 250 values for the red entries (below). Mean and standard deviation are plotted across 10 seeds, maximum entry is highlighted in black.}
    \label{fig:2good2bad_snapshot}
\end{figure}

\begin{figure}[h!]
    \centering
    \includegraphics[width=0.32\textwidth]{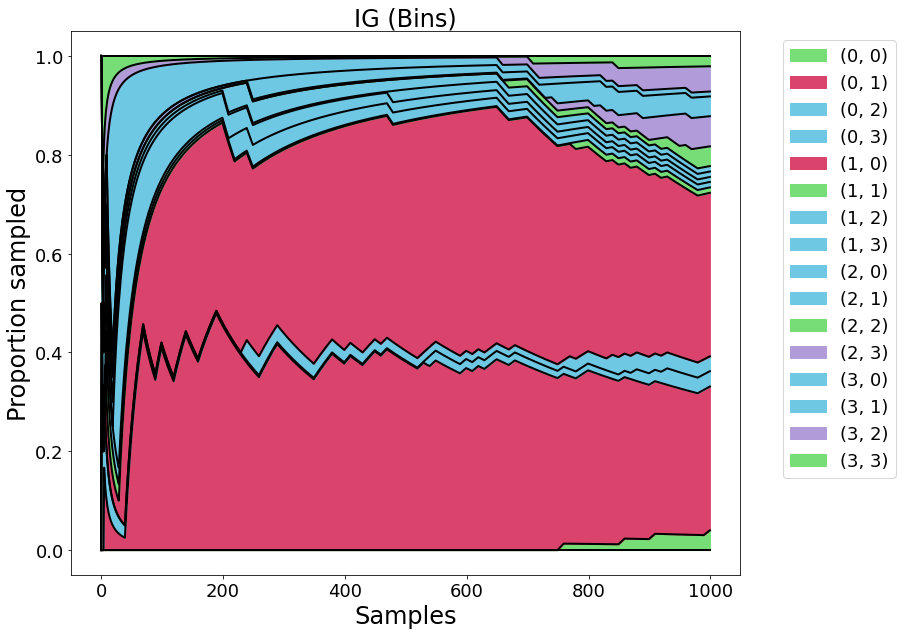}
    \includegraphics[width=0.32\textwidth]{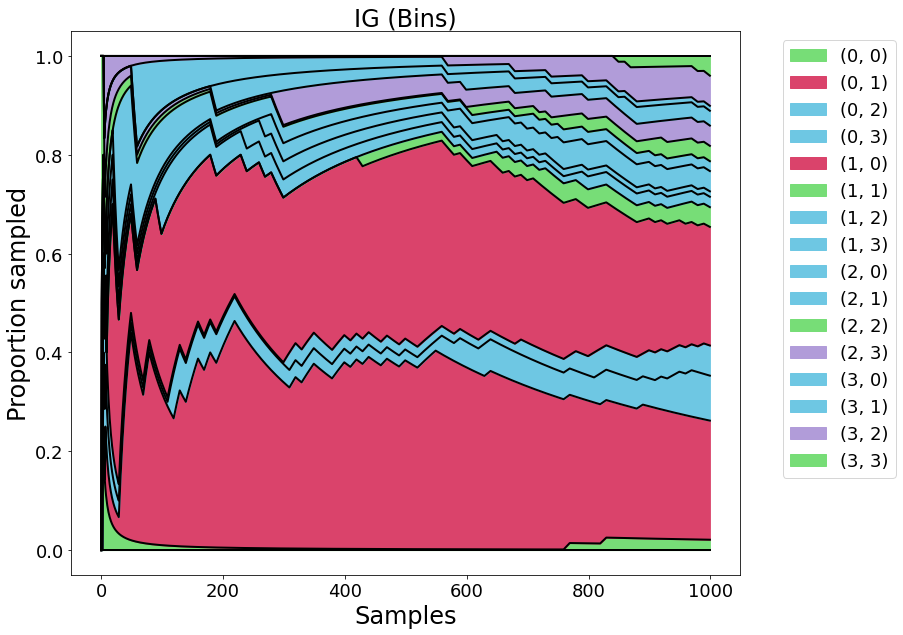}    
    \includegraphics[width=0.32\textwidth]{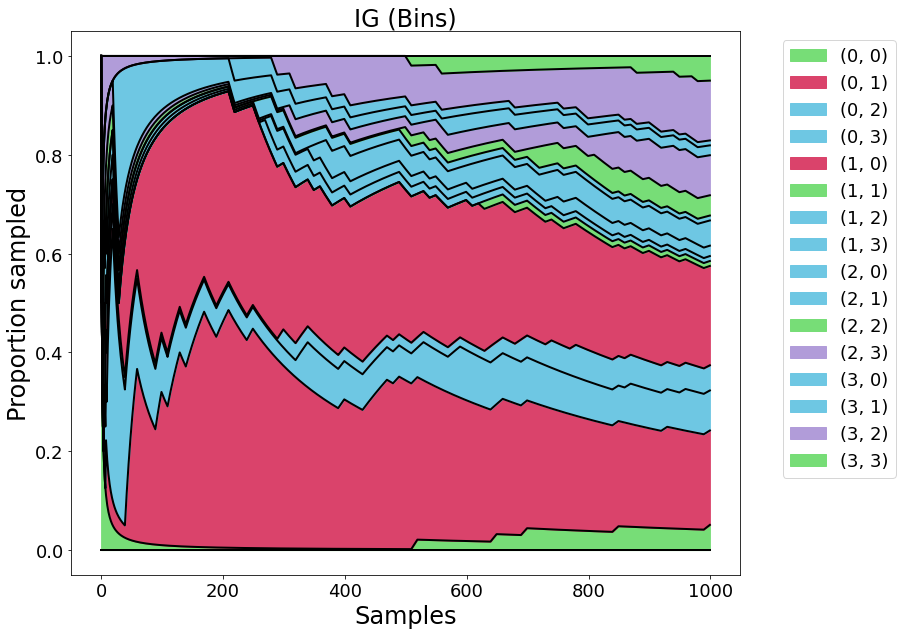}
       
    \includegraphics[width=0.32\textwidth]{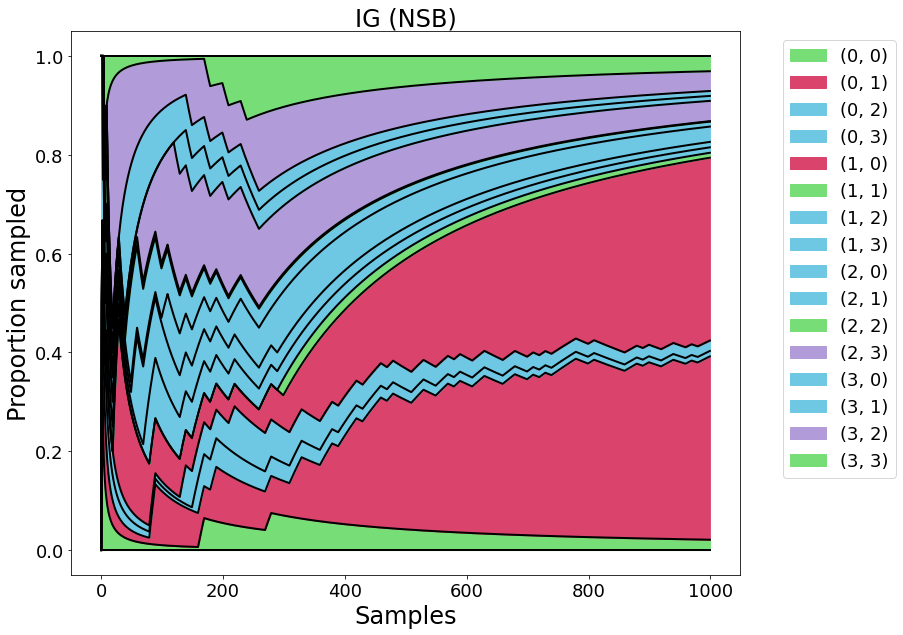}
    \includegraphics[width=0.32\textwidth]{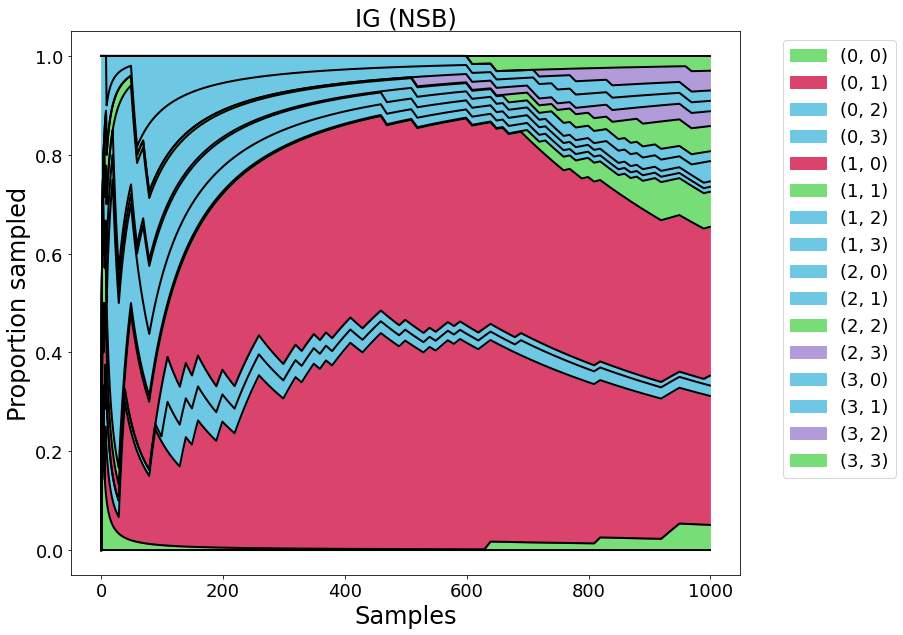}
    \includegraphics[width=0.32\textwidth]{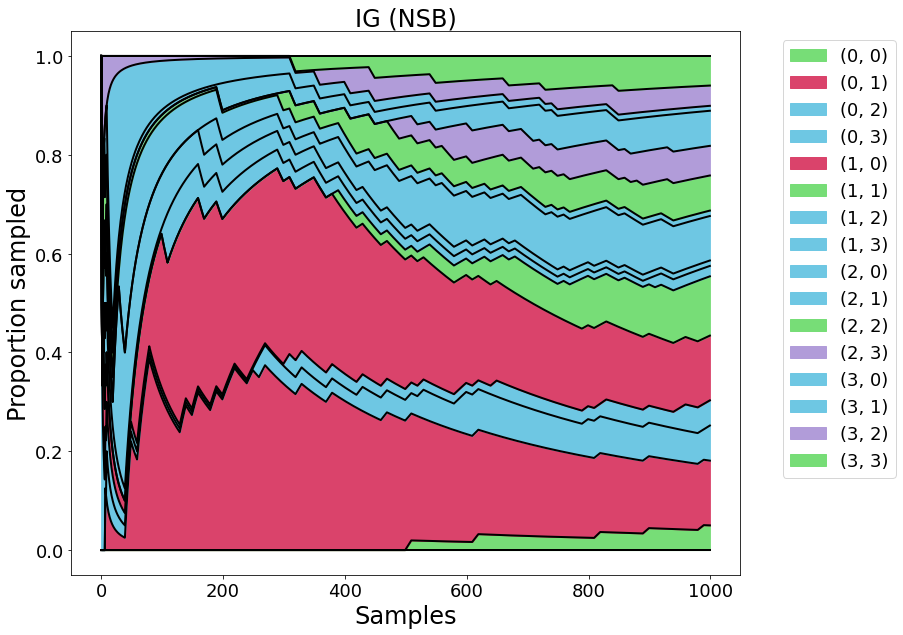}
    
    \includegraphics[width=0.32\textwidth]{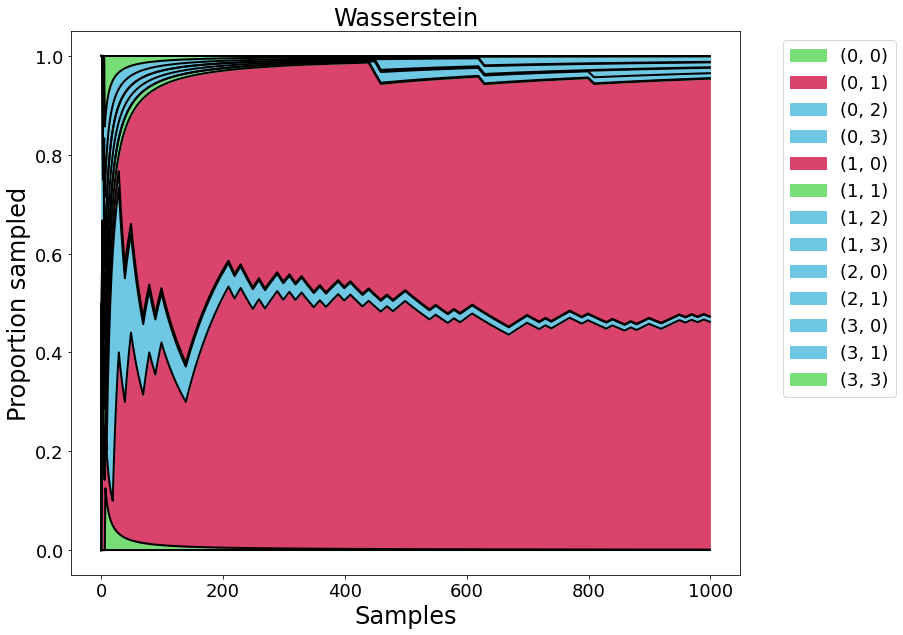}
    \includegraphics[width=0.32\textwidth]{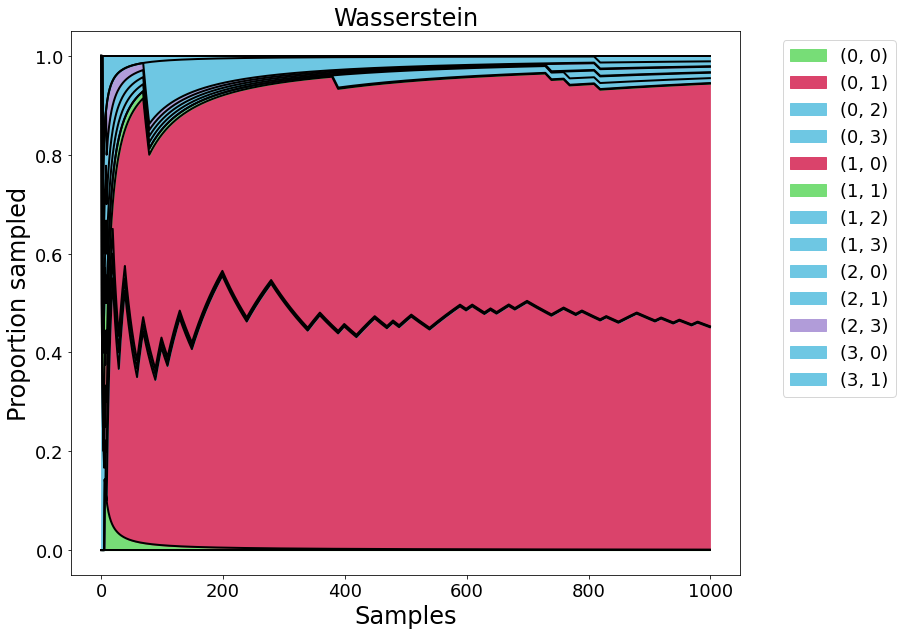}
    \includegraphics[width=0.32\textwidth]{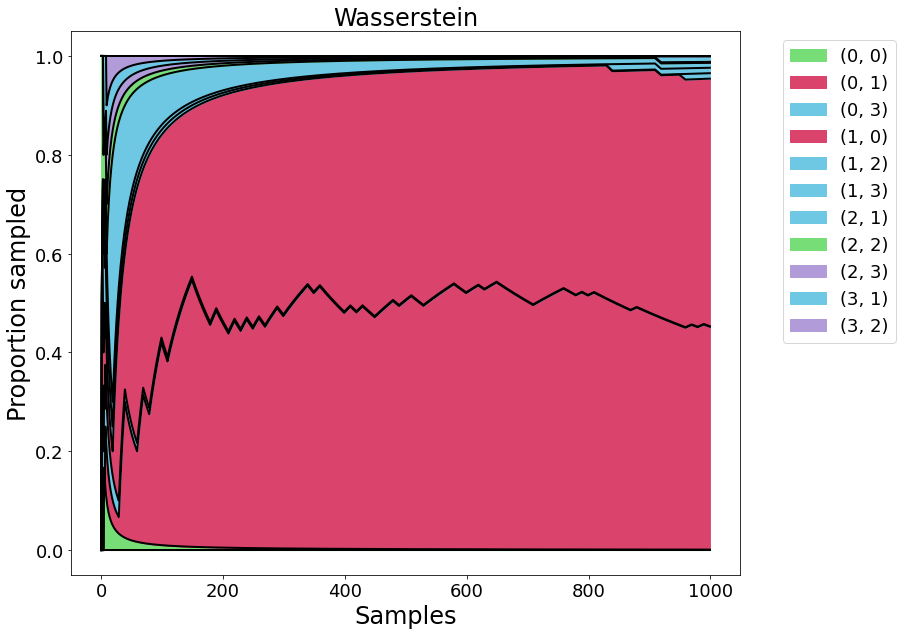}
    
    \includegraphics[width=0.32\textwidth]{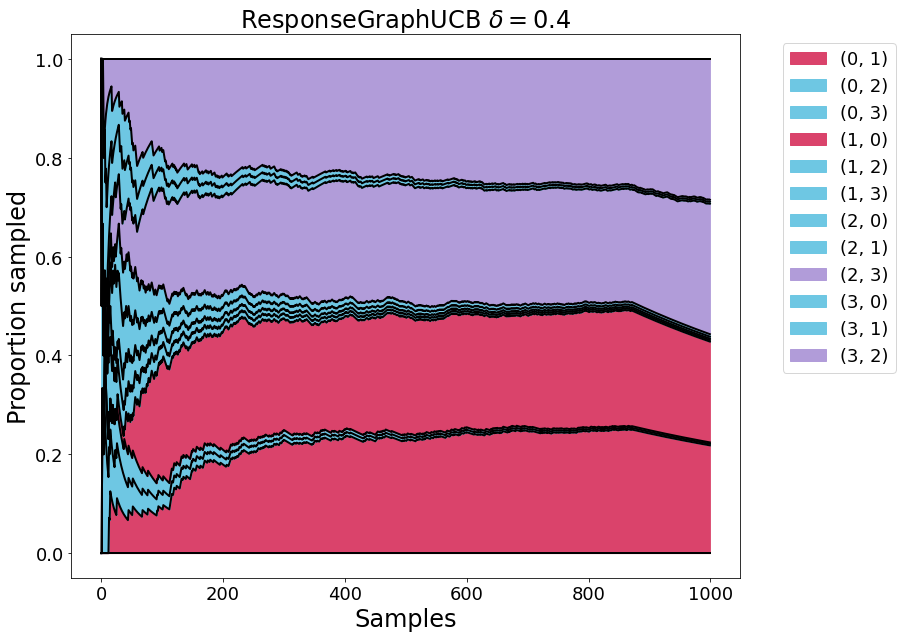}
    \includegraphics[width=0.32\textwidth]{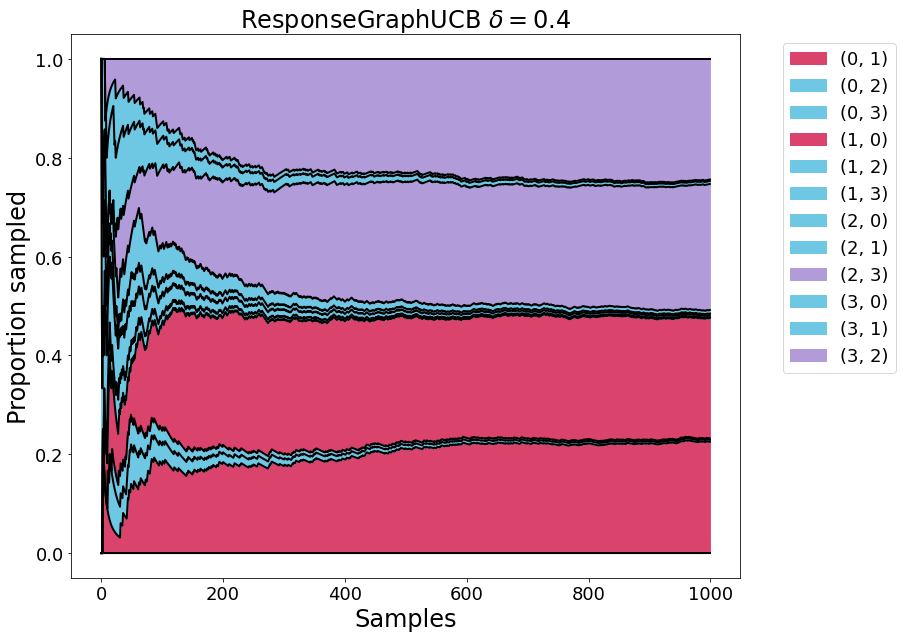}
    \includegraphics[width=0.32\textwidth]{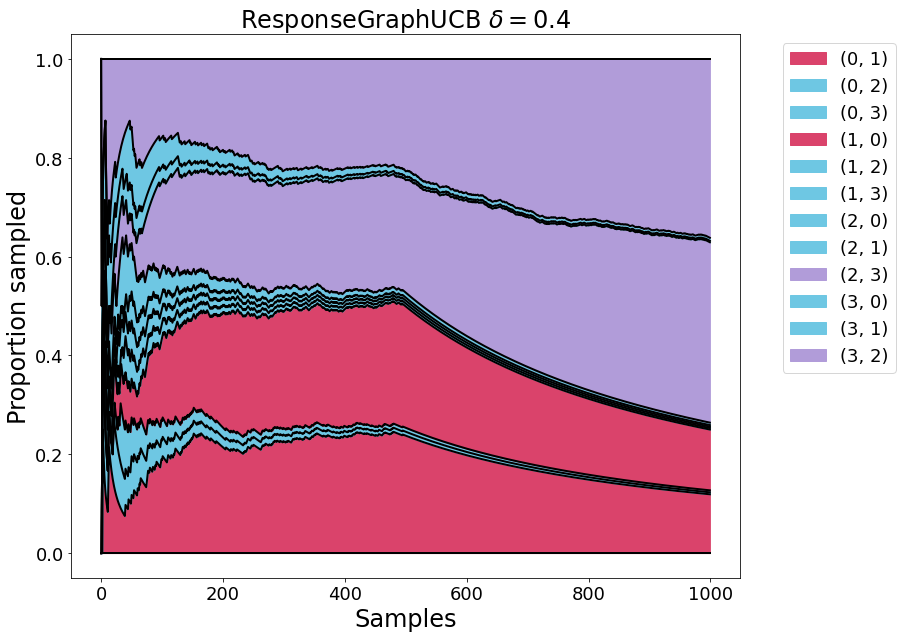}
    
    \caption{Proportion of entries sampled on 2 Good, 2 Bad for more seeds.}
    \label{fig:2good_2bad_entries_all}
\end{figure}

\subsection{3 Good, 3 Bad}

\begin{figure}[h!]
    \centering
    \includegraphics[width=0.45\textwidth]{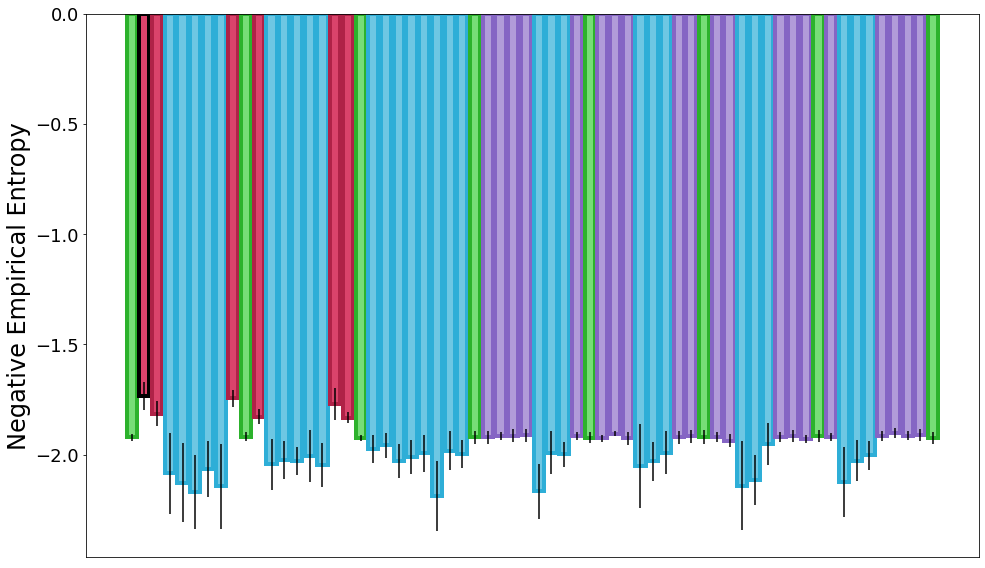}
    \includegraphics[width=0.45\textwidth]{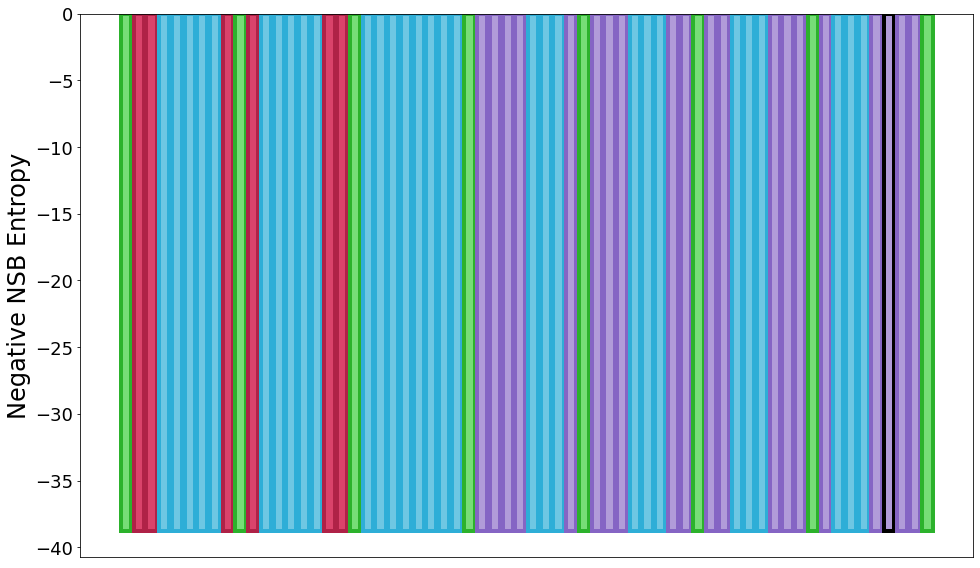}
    \includegraphics[width=0.45\textwidth]{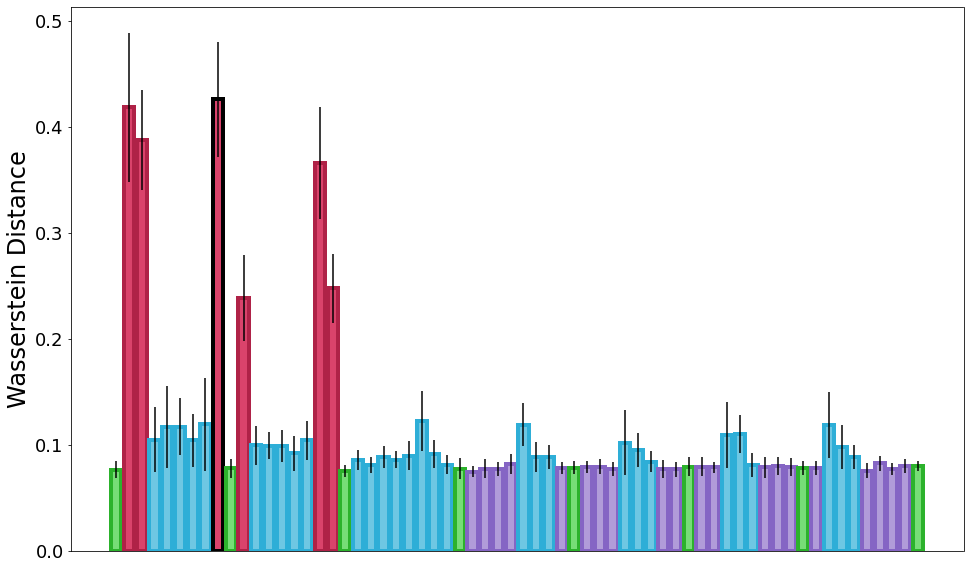}
    \caption{Values of the objectives for each entry on `3 Good, 5 Bad' after sampling 5 values for every entry. Mean and standard deviation across 10 seeds is shown, maximum highlighted in black.}
    \label{fig:3good_5_5}
\end{figure}

\begin{figure}[h!]
    \centering
    \includegraphics[width=0.45\textwidth]{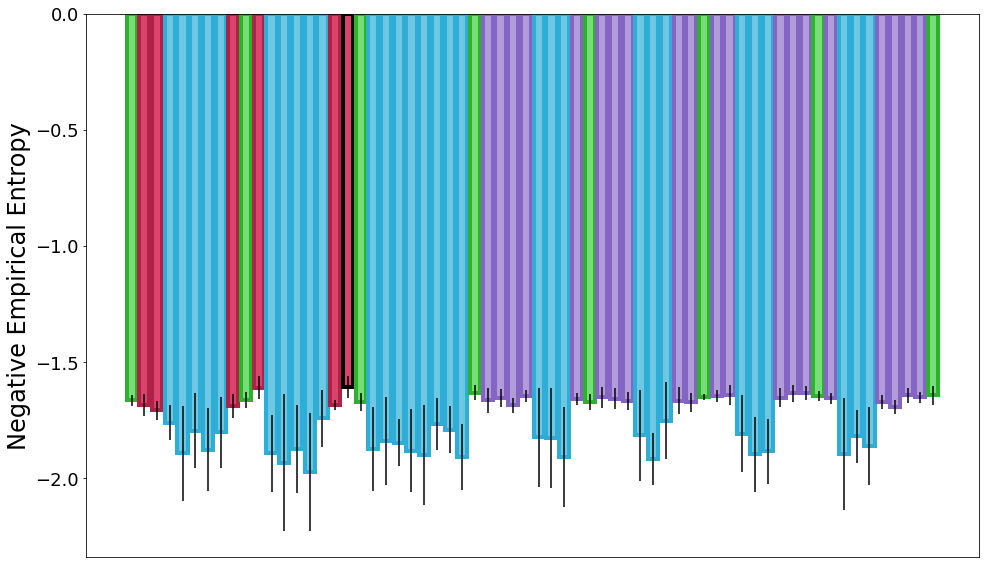}
    \includegraphics[width=0.45\textwidth]{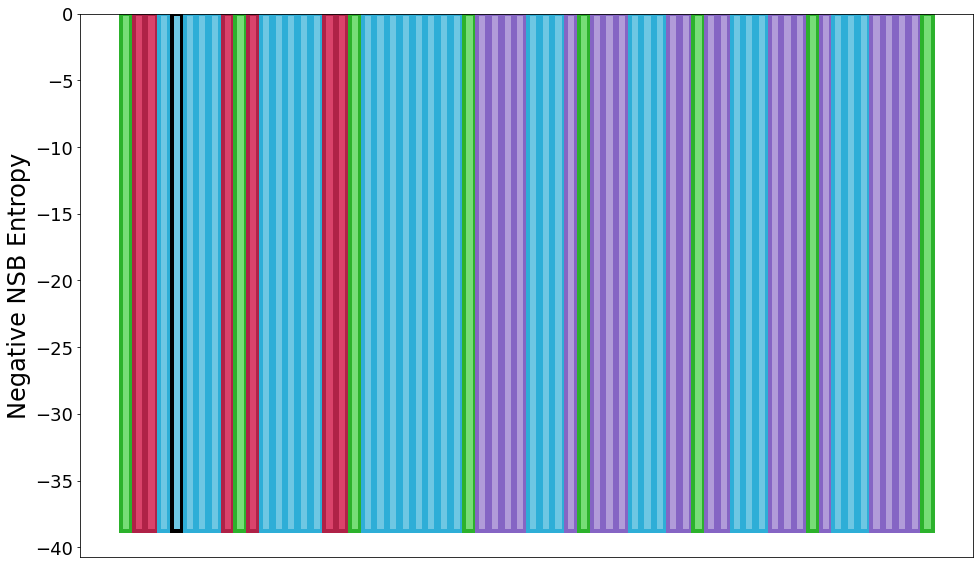}
    \includegraphics[width=0.45\textwidth]{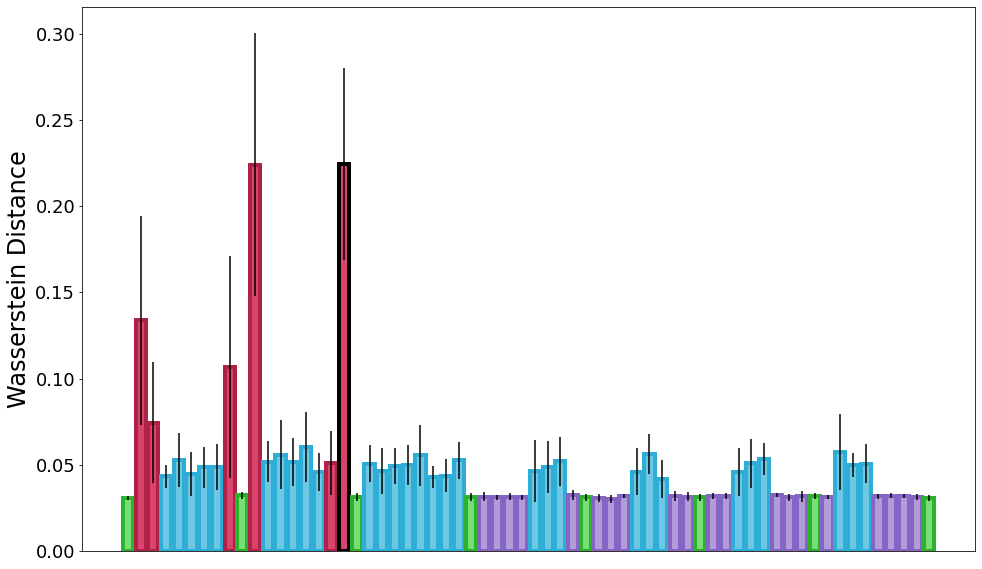}
    \caption{Values of the objectives for each entry on `3 Good, 5 Bad' after sampling 5 values for every entry, and additionally sampling 1000 values for the red entries. Mean and standard deviation across 10 seeds is shown, maximum highlighted in black.}
    \label{fig:3good_5_1005}
\end{figure}

\begin{figure}[h!]
    \centering
    \includegraphics[width=0.32\textwidth]{images/3good_5bad/entries/entries_ent_0_1.png}
    \includegraphics[width=0.32\textwidth]{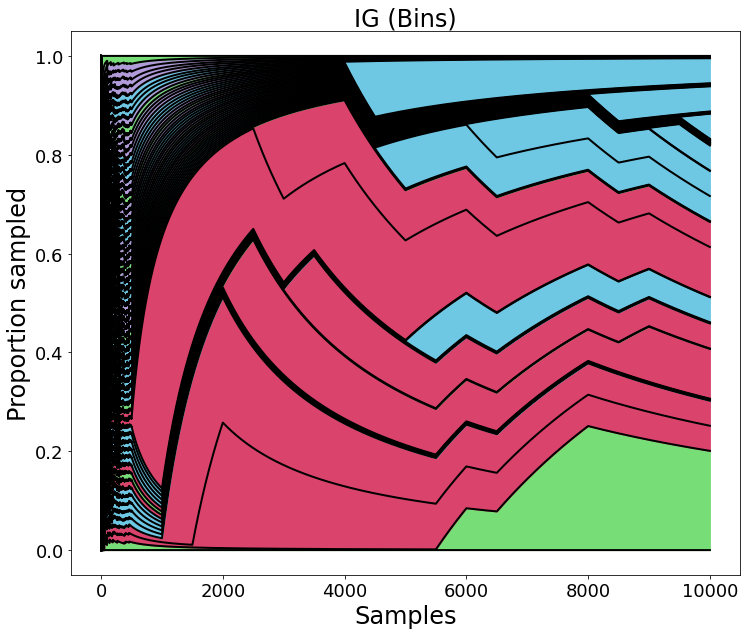}     \includegraphics[width=0.32\textwidth]{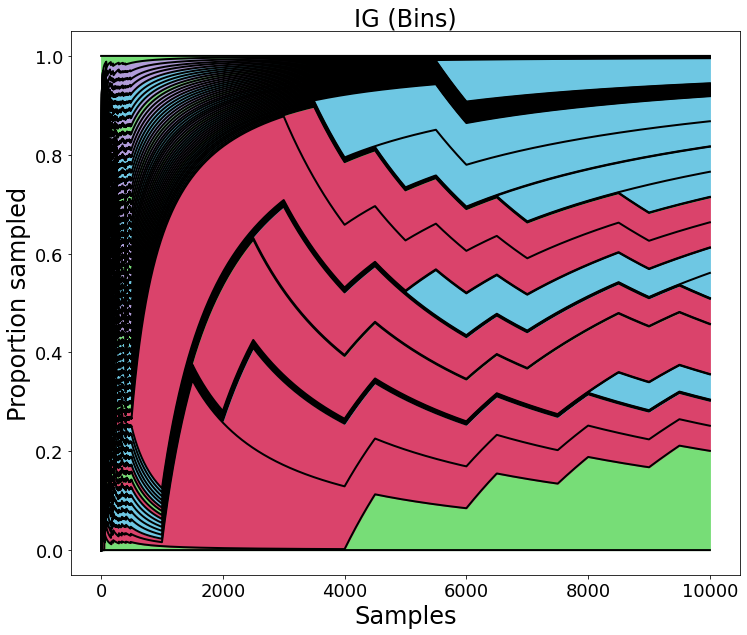}
       
    \includegraphics[width=0.32\textwidth]{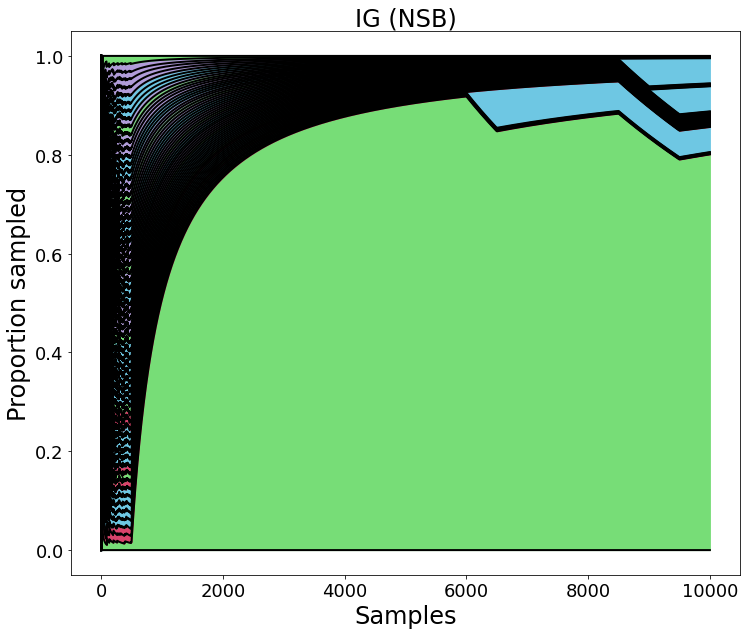}
    \includegraphics[width=0.32\textwidth]{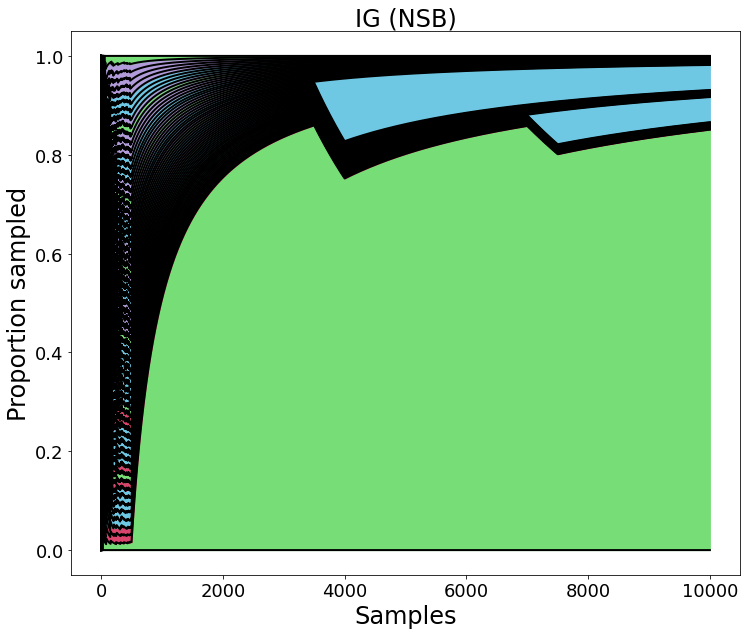}     \includegraphics[width=0.32\textwidth]{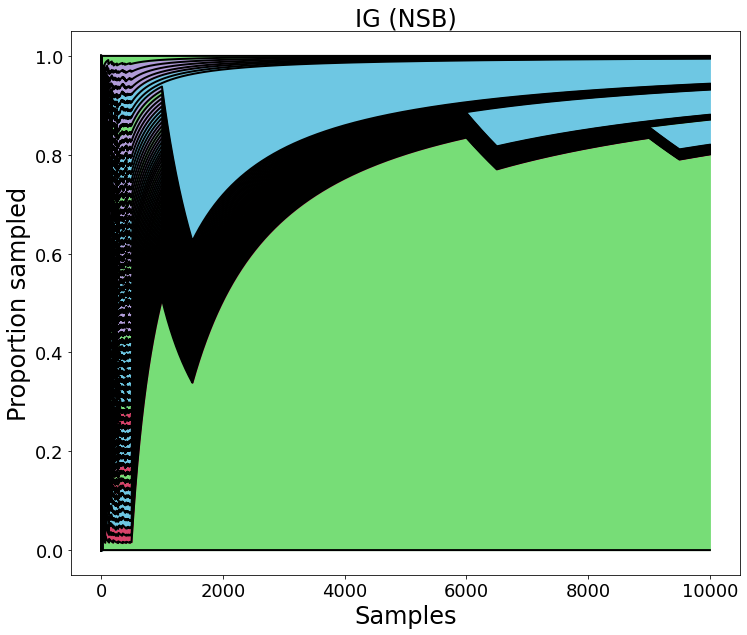}
    
    \includegraphics[width=0.32\textwidth]{images/3good_5bad/entries/entries_wass_0_1.png}
    \includegraphics[width=0.32\textwidth]{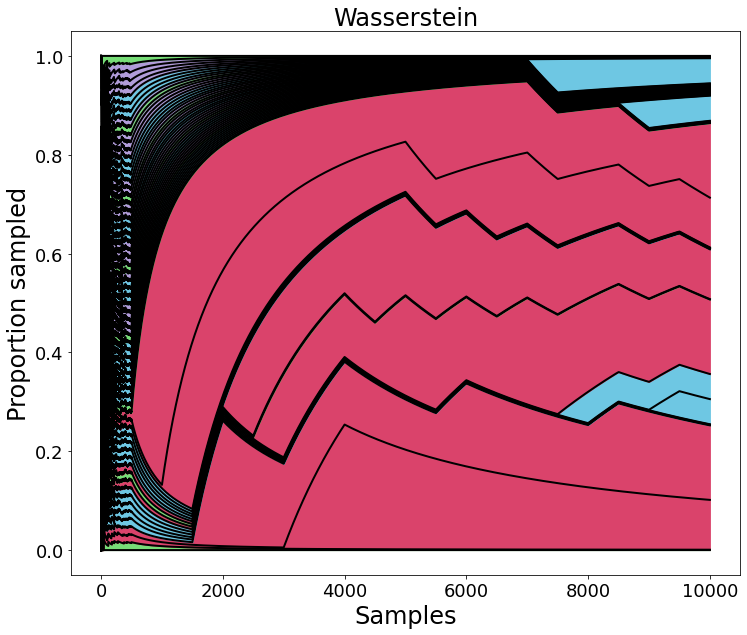}    \includegraphics[width=0.32\textwidth]{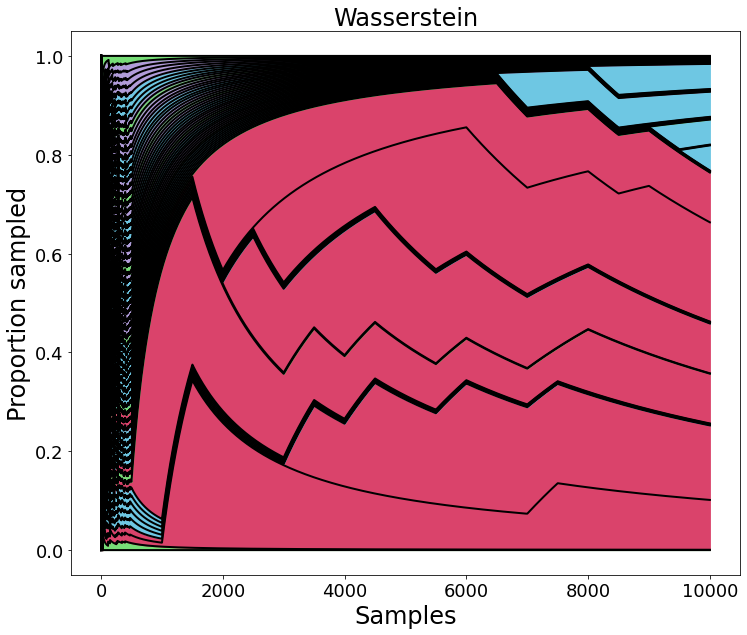}
    
    \includegraphics[width=0.32\textwidth]{images/3good_5bad/entries/entries_f_1.png}
    \includegraphics[width=0.32\textwidth]{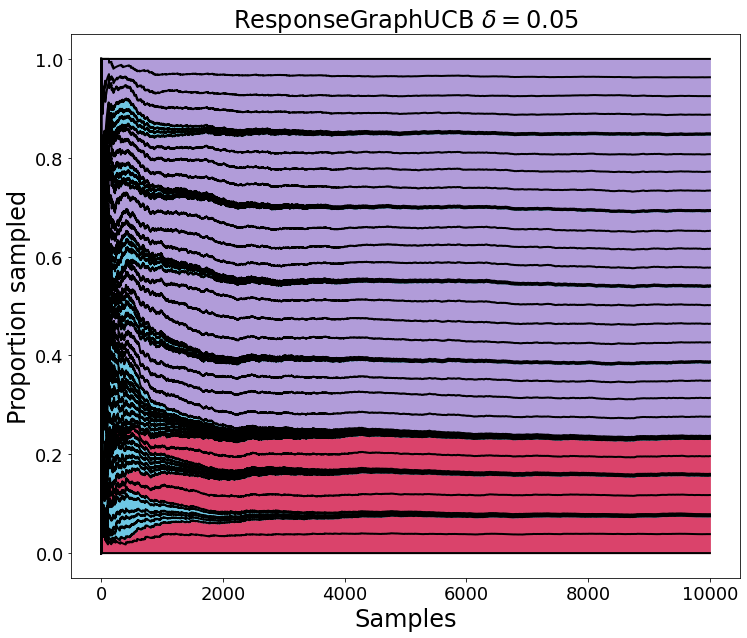}     \includegraphics[width=0.32\textwidth]{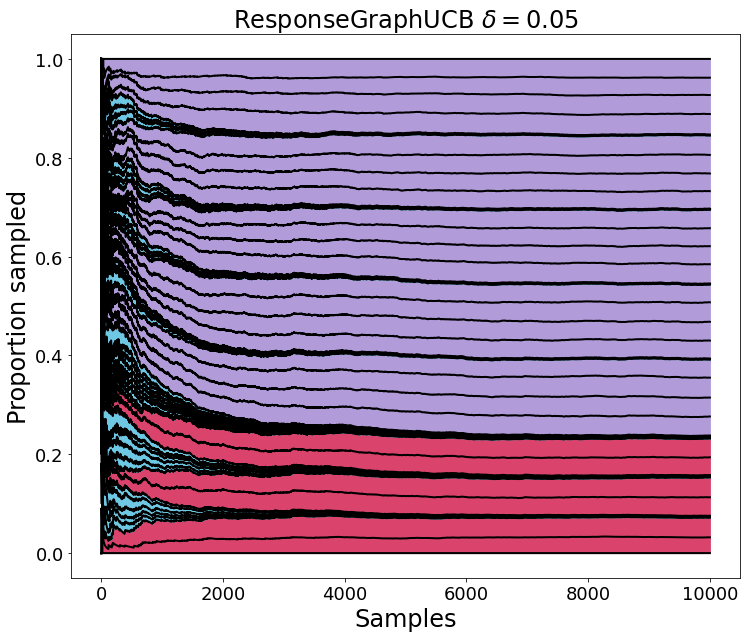}
    
    \caption{Proportion of entries sampled on 3 Good, 5 Bad for more seeds.}
    \label{fig:3good_5bad_entries_all}
\end{figure}
\clearpage
\section{Proofs}
\label{sec:proofs}

\paragraph{Notation}
As a reminder, we reintroduce notation that is relevant to this section.\\
$M_\star$ is the true payoff vector, which is unknown to us.\\
$M$ is our prior distribution over the entries of the payoff vector, represented by a GP.\\
The GP models noise in the observations of the payoff as $\tilde{M} = M + \epsilon$, where $\epsilon \sim \mathcal{N}(0,I\sigma_A^2)$.
We model the payoffs we receive from the real game at timestep $t$ when taking action $a_t$, as $m_t \sim M_\star(a_t) + \epsilon'_t$. Where, $\epsilon'_t$ is i.i.d. and has support on the interval $[-\sigma_A, \sigma_A]$.

\paragraph{Clipped Noise}
Note that it is important the observation noise $\epsilon'_t$ is clipped, since it allows us to apply Lemma \ref{lemma-gp-concentration} which is an existing result from \citet[Theorem 6]{srinivas2009gaussian}. In that paper, \citet{srinivas2009gaussian} assume the noise terms $\epsilon'_t$ are uniformly bounded by $\sigma_A$ which is equivalent to all $\epsilon'_t$ having support on the interval $[-\sigma_A, \sigma_A]$.
Since this assumption is shared by the seminal paper \citep{srinivas2009gaussian}, we do not believe it to be overly restrictive in our theoretical analysis.

\paragraph{Regret}

We quantify the performance of our method by measuring regret. Our main analysis relies on Bayesian regret \citep{information-directed-informs}, defined as
\begin{gather}
    \label{eq-regret-b}
    J_t^B = 1 - \expect{}{P(\optrv = \bestGuess \vert H_t = h_t) }{}, 
\end{gather}
where the expectation is taken over the following:
\begin{itemize}
    \item Our prior distribution, $M_t$, representing our uncertainty over the true unknown payoff vector $M_\star$ at timestep $t$.
    \item The randomness in the actions we have taken and the corresponding observations we have received up to timestep $t$. These are encoded by our history $H_t$, in which a particular realization is $h_t = a_1, m_1, ... ,a_t, m_t$. $a_t \sim A_t$, our distribution over actions to take at timestep $t$, and $m_t \sim M_\star(a_t) + \epsilon'_t$ our clipped noise model when interacting with the game.
\end{itemize}
In this formulation $\bestGuess$ is used to denote the $\alpha$-rank with the highest probability under $r$ at time $t$, where
$r$ is the distribution over $\alpha$-ranks according to the prior, $P(r) = P(M_t \in f^{-1}(r))$.
Since $J_t^B$, like all purely Bayesian notions, does not involve the ground truth payoff, we need to justify its practical relevance. We do this by benchmarking it against two notions of frequentist regret. 
The first measures how accurate the probability we assign to the ground truth $\groundTruth = f(\groundTruthM)$ is
\begin{gather}
    J_t^F = 1 - \expect{h_t}{P(\optrv = \groundTruth \vert H_t = h_t) }{}.
\end{gather}
The second measures if the mean of our payoff belief, which we denote $M_\mu$, evaluates to the correct $\alpha$-rank
\begin{gather}
    J_t^M = 1 - \expect{h_t}{\idone{ f(M_\mu) = \groundTruth}}{},
\end{gather}
where the symbol $\idone{\text{predicate}}$ evaluates to 1 or 0 depending on whether the predicate is true or false. 
For both these notions of regret the expectation is taken only over the history $h_t \sim H_t$.

\paragraph{Permutation Property}
We begin by explicitly stating a property of the infinite-$\alpha$ version of $\alpha$-rank. 
The function $f$ computing the $\alpha$-rank satisfies the permutation property, defined as   
\begin{gather}
\label{permutation-property}
\pi(M_1) = \pi(M_2) \;\; \Longrightarrow \;\; f(M_1) = f(M_2).
\end{gather}
Here, $\pi(M)$ denotes the ordering of the elements of the vector $M$ using the standard $\geq$ operation on real numbers. This is the same property exploited by frequentist analysis by \citet{rowland2019multiagent}. 
Letting $R := \mathcal{R}^{|\mathcal{S}|}$, be the space of all valid $\alpha$-ranks,
Property \eqref{permutation-property} implies that $R$ is a finite set and 
\begin{gather}
|R| \leq N!,
\label{number-of-alpha-ranks-bounded}
\end{gather}
where $N := |\mathcal{S}|$ the number of pure strategies/actions.
Note that our proofs consider the general multi-population case of $\alpha$-rank, and our not restricted to just the single population scenario.

\paragraph{Separability Assumption} Similarly to the work of \citet{rowland2019multiagent}, we limit ourselves to payoffs that are distinguishable in order to make $\alpha$-rank robust to small changes in the payoffs. We assume that there exists a constant $\Delta > 0$ such that for all payoff indices $i$, $j$
\begin{gather}
\label{eq-separable}
\left \vert \groundTruthM(i) - \groundTruthM(j) \right \vert \geq \Delta.
\end{gather}

\paragraph{Information Gain and Entropy}
We recall a formula for the information gain in terms of the entropy:
\begin{align}
    \informationGain{r}{(\tilde{M}_t(a),a)}{H_t = h_t} &= \entropy{r}{H_t = h_t} - \entropy{r}{H_t = h_t, A_t = a, \tilde{M}_t(a)} \\
    &=\entropy{r}{H_t = h_t} - \mathop{\E}_{\tilde{m}_t \sim \tilde{M}_t(a)} \left[\entropy{r}{H_t = h_t, A_t=a, \tilde{M}_t(a)=\tilde{m}_t}\right].
    \label{ig-as-entropy}
\end{align}

\paragraph{Regret Bound For Policy Maximizing Information Gain on Payoffs}
We now show a regret bound for a policy that maximizes information gain on the payoffs. 
Define:
\begin{gather}
    \pi_{\text{IGM}}(a \vert H_t = h_t) = \argmax_a \informationGain{\tilde{M}_t}{(\tilde{M}_t(a),a)}{H_t = h_t},
    \label{eq-policy-ig-m-appendix}
\end{gather}
as the policy which selects the action that maximises the information gain on the payoffs (given any history $H_t$ and prior $M_t$).
Let $H_T^{\text{IGM}}$ denote the history when following $\pi_{\text{IGM}}$ for $T$ timesteps.\\

{\bf Proposition 1} [Regret Bound For Information Gain on Payoffs]
If we select actions using strategy $\pi_{ \text{IGM} }$, the regret at timestep $T$ is bounded as
\begin{gather}
J^B_T \leq J^F_T \leq 1 - \expect{h_T \sim H_T^{\text{IGM}}}{P(\optrv = \groundTruth \vert H_T = h_T) }{} \leq 
        { T e^{g(T)} }  \;\; \text{where} \;\; 
        g(T) = \mathcal{O}(-\sqrt[3]{\Delta^2 T} )
        .
\end{gather}

\begin{proof}

We know that $P(r=r_{\star} | H_T = h_T) \geq P(\optrv = \groundTruth \vert H_T = h_T)$ since $r_{\star}$ is defined as the $\alpha$-rank with the highest probability under $r$ and time $t$.\\
Thus, $J^B_T \leq J^F_T$ for any history. 

Fix a history $h_T$. By assumption of separability, we have
\begin{gather}
P(\optrv = \groundTruth \vert H_T = h_t) \geq
P \left(\vert M_t - M_\star \vert_\infty \leq \frac{\Delta}{2}  \right).
\end{gather}
We now use concentration results for Gaussian Processes. Specifically, we invoke Corollary \ref{cor-sample}, stated later, together with an explicit formula for $g(T)$.

This proves $J^F_T \leq T e^{g(T)}$, ending our proof.
\end{proof}

\paragraph{Regret Bound For Policy Maximizing Information Gain on $\alpha$-Ranks}
We move on to show a bound for a policy that maximizes information gain on the $\alpha$-ranks.
Define:
\begin{gather}
    \pi_{\text{IGR}} = \argmax_{a_1,\dots,a_T} \informationGain{r}{(\tilde{M}_1(a_1),a_1), \dots, (\tilde{M}_T(a_T),a_T)}{},
    \label{eq-policy-ig-r-appendix}
\end{gather}
as the policy which maximizes information gain on the $\alpha$-ranks directly. 
Note that this is an extension of \eqref{eq-ig-entropy} to $T$-step look-ahead.

Let $H_T^{\text{IGR}}$ denote the history when following $\pi_{\text{IGR}}$ up to timestep $T$.
Denote by $h_b(p) = -(p\log p + (1-p) \log (1-p))$ the entropy of a Bernoulli random variable with parameter $p$, and denote $h_b^{-1}$ as the inverse of the restriction of $h_b$ to the interval $[1/2,1]$.\\

{\bf Proposition 2 Expanded~}[Regret Bound For Information Gain on Belief over $\alpha$-Ranks]
If we select actions using strategy $\pi_{ \text{IGR} }$, the bayesian regret is bounded as
\begin{gather}
J^B_T = 1 - \expect{h_T \sim H_T^{\text{IGR}}}{P(\optrv = \bestGuess \vert H_T = h_T)}{} \leq \nonumber\\ 
1 - \idone{T e^{g(T)} \leq (N \log N)^{-1}} \bigg[ h_b^{-1}(\expect{h_T \sim H_T^{\text{IGM}}}{ h_b(1 - Te^{g(T)}) + T e^{g(T)} N \log N}{}) \bigg],
\end{gather}
and $g(T)$ is as in Proposition \ref{proposition-regret-bound-payoffs}.

\begin{proof}
We start by bounding the entropy of the $\alpha$-rank distribution. 
Let the abbreviation $p^\star = P(\optrv = \bestGuess \vert H_T = H_T^{\text{IGM}})$.  

We have
\begin{align*}
    \entropy{r}{H_T^{\text{IGR}}} \hspace{-1.5cm}& \\
    & \overset{(a)}{\leq} \entropy{r}{H_T^{\text{IGM}}} \\ 
    & \overset{(b)}{\leq} \expect{h_T \sim H_T^{\text{IGM}}}{ h_b( p^\star ) + (1 - p^\star) \log(\vert R  \vert)}{} \\ 
    & \overset{(c)}{\leq} \expect{h_T \sim H_T^{\text{IGM}}}{ h_b( p^\star ) + (1 - p^\star) N \log N}{}.
\end{align*}
Here, (a) follows from the definition of $\pi_{\text{IGR}}$ and equation \eqref{ig-as-entropy}, (b) follows by Lemma \ref{lemma-entropy-ub} and (c) holds because $\vert R \vert \leq N!$ by Equation \eqref{number-of-alpha-ranks-bounded}. 
Combining the above with the bound $1 - p^\star \leq T e^{g(T)}$ from Proposition \ref{proposition-regret-bound-payoffs}, we have
\begin{align}
     \entropy{r}{H_T^{\text{IGR}}} &\leq \expect{h_T \sim H_T^{\text{IGM}}}{h_b( p^\star ) + (1 - p^\star) N \log N}{} \\
     &\leq \expect{h_T \sim H_T^{\text{IGM}}}{h_b( p^\star ) + T e^{g(T)} N \log N}{}.
\label{bound-entropy-rank}
\end{align}

Let us now assume that $T e^{g(T)} \leq 1/2$, since we are interested in the behaviour of our regret bound as $T \to \infty$, and we know that as $T \to \infty, T e^{g(T)} \to 0$.
If $T e^{g(T)} > 1/2$ then we can trivially bound our expression above by $1$.
Then $p^\star \geq 1 - T e^{g(T)} \geq 1/2 \implies h_b(p^\star) \leq h_b(1 - Te^{g(T)})$.

We now proceed to bound the probability of $\bestGuess$ in terms of the entropy of the $\alpha$-ranks. We have
\begin{gather*}
h_b( P(\optrv = \bestGuess \vert H_T^{\text{IGR}} = h_t) ) \leq \entropy{r}{H_T^{\text{IGR}}}.
\end{gather*}
This, together with \eqref{bound-entropy-rank} and $h_b(p^\star) \leq h_b(1 - Te^{g(T)})$ implies 
\begin{align}
h_b( P(\optrv = \bestGuess \vert H_T^{\text{IGR}} = h_t) ) \leq \expect{h_T \sim H_T^{\text{IGM}}}{h_b(1 - Te^{g(T)}) + T e^{g(T)} N \log N}{}.
\end{align}

Since the codomain of $h_b$ is $[0,1]$, we must introduce additional restrictions in order to be able to invert the function.

To ensure $h_b(1 - Te^{g(T)}) + T e^{g(T)} N \log N \leq 1$ we restrict our analysis to when $T e^{g(T)} \leq (N \log N)^{-1}$. 
Note that this subsumes our earlier restriction of $T e^{g(T)} \leq 1/2$.
Again, we can trivially bound our final expression above by 1, should this condition not be met.

We denote by $h_b^{-1}$ the inverse of the restriction of $h_b$ to the interval $[1/2,1]$.
Note that $h_b(x) \leq y \implies x \geq h_b^{-1}(y)$ for $x \in [1/2, 1], y \in [0,1]$.

\begin{gather}
     h_b( P(\optrv = \bestGuess \vert H_T^{\text{IGR}} = h_t) ) \leq  \expect{h_T \sim H_T^{\text{IGM}}}{h_b(1 - Te^{g(T)}) + T e^{g(T)} N \log N}{}
\end{gather}
$\implies$\\
\begin{gather}
     P(\optrv = \bestGuess \vert H_T^{\text{IGR}} = h_t) \geq h_b^{-1}(\expect{h_T \sim H_T^{\text{IGM}}}{ h_b(1 - Te^{g(T)}) + T e^{g(T)} N \log N}{}) 
\end{gather}
$\implies$\\
\begin{gather}
     1 - P(\optrv = \bestGuess \vert H_T^{\text{IGR}} = h_t) \leq 1 - h_b^{-1}(\expect{h_T \sim H_T^{\text{IGM}}}{ h_b(1 - Te^{g(T)}) + T e^{g(T)} N \log N}{}).
\end{gather}

Finally, we state our final regret bound incorporating our restrictions/assumptions we have made.

\begin{gather}
    1 - P(\optrv = \bestGuess \vert H_T^{\text{IGR}} = h_t) \leq 1 - \idone{T e^{g(T)} \leq (N \log N)^{-1}} \bigg[ h_b^{-1}(\expect{h_T \sim H_T^{\text{IGM}}}{ h_b(1 - Te^{g(T)}) + T e^{g(T)} N \log N}{}) \bigg].
\end{gather}

This then proves the expanded form of the proposition.

The simpler form of the Proposition in Section \ref{sec-theory} then follows.
This is because as $T \to \infty, T e^{g(T)} \to 0$ and both $h_b(1 - Te^{g(T)}), T e^{g(T)} N \log N \to 0$ thus ensuring that $h_b^{-1}( h_b(1 - Te^{g(T)}) + T e^{g(T)} N \log N) \to 1$.

\end{proof}

We use the following result by \citet[their Theorem 6]{srinivas2009gaussian}, which we specialize in our notation. We use the term Gaussian Process despite the fact that the index set is finite, since the model includes observation noise. 

\begin{lemma}[\citeauthor{srinivas2009gaussian}, Concentration for a Gaussian Process]
\label{lemma-gp-concentration}
Consider a Gaussian Process $M$, with $N$ indices. Assume M uses a zero-mean prior with constant variance $\sigma_0^2$ and observation noise $\sigma_A$. The posterior process $M_t$ is obtained by conditioning on $t$ observations. The observations are obtained as $m_t = M_\star(a_t) + \epsilon'_t$, where $ \epsilon'_t $ are i.i.d random variables with support bounded by $[-\sigma_0, \sigma_0]$. Denote the RKHS norm of $M_\star$ under the GP prior with $ \| M_\star \|_k$. Denote by $\gamma^\star_t$ the maximum information gain about $M$ obtainable in $t$ timesteps. Then, for any $\Delta > 0$, and for any timestep $t$, we have
\begin{gather}
P \left[\vert M - M_\star \vert_\infty \leq \frac{\Delta}{2}  \right] \geq 1 - t e^{-\sqrt[3]{\frac{\left( \frac{\Delta}{2 \sigma_t^{ \text{max} } }  \right)^2-2 \| M_\star \|_k}{300 \gamma^\star_t}}}.
\end{gather}

\end{lemma}

The above lemma requires knowledge of the RKHS norm and the maximum obtainable information gain.

\begin{lemma}[Worst-Case Constants]
\label{lemma-cosntants}
For any  kernel, we have
\begin{gather*}
\| M_\star \|_k \leq \frac1{\sigma_0^{-2}} \| M_\star \|_2^2 \;\; \text{and} \;\; \gamma^\star_t \leq \frac12 \log \det(I + \sigma_A^{-2} K). 
\end{gather*}
Moreover, for a strategy that maximizes information gain on payoffs, we have
\begin{gather*}
\textstyle \sigma_t^{ \text{max} } \leq \frac{ \sigma_A  \sigma_0}{ \sqrt{ \sigma_A^2 + (\frac{T}{N} - 1) \sigma_0^2} }.
\end{gather*} 
\end{lemma}
\begin{proof}
The inequalities for posterior variance and the RHKS norm are obtained by using the independent kernel, which represents the worst-case. The inequality for information gain follows by writing
\begin{gather}
\gamma^\star_t = \frac12 \log \frac{\det(I + \sigma_A^{-2} K)}{\det(I + \sigma_A^{-2} \Sigma)} \leq \frac12 \log \det(I + \sigma_A^{-2} K). 
\end{gather}
The inequality follows since the denominator is greater than one. Here, we denoted the prior covariance with $K$ and the posterior covariance with $\Sigma$. 
\end{proof}

\begin{corollary}
\label{cor-sample}
For a strategy that maximizes the payoff information gain and for any time-step $T$, we have: 
\begin{gather}
P \left[\vert M_t - M_\star \vert_\infty \leq \frac{\Delta}{2}  \right] \geq 1 - T e^{g(T)}, \;\; \text{where} \;\; g(T) = \mathcal{O}(-\sqrt[3]{\Delta^2 T} )
\nonumber
\end{gather}
Specifically,
\begin{gather}
    g(T) = -\sqrt[3]{\frac{\left( \frac{\Delta}{2} \right)^2 
    \frac{  \sigma_A^2 + (\frac{T}{N} - 1) \sigma_0^2} { \sigma^2_A  \sigma^2_0}
        -2 \frac1{\sigma_0^{-2}} \| M_\star \|_2^2 }{300 \frac12 \log \det(I + \sigma_A^{-2} K) }}.
\nonumber
\end{gather}
\end{corollary}

\begin{lemma}[Upper Bound on Entropy]
\label{lemma-entropy-ub}
For any  discrete random variable $x$ with $n$ outcomes, we have, for each outcome $i$
\begin{gather*}
 \entropy{x}{} \leq h_b( p_i ) + (1 - p_i) \log( n - 1).
\end{gather*}
\end{lemma}
\begin{proof}
Without loss of generality, assume $i=1$.
\begin{align*}
\entropy{x}{} 
&= -p_1 \log p_1 - \sum_{j > 1} p_j \log(p_j) \\
&= -p_1 \log p_1 - (n-1) \sum_{j > 1} \frac{1}{n-1} p_1 \log(p_j) \\
& \overset{(a)}{\leq} -p_1 \log p_1 - (n-1) \left( \sum_{j > 1} \frac{p_j}{n-1} \right) \log \left( \sum_{j > 1} \frac{p_j}{n-1} \right)  \\
&= -p_1 \log p_1 - (1 - p_1) \log \left(  \frac{1 - p_1}{n-1} \right)  \\
&= -p_1 \log p_1 - (1 - p_1) \log \left(  \frac{1 - p_1}{n-1} \right)  \\
&= h_b( p_j ) + (1 - p_j) \log( n - 1 )
\end{align*}
There, (a) follows from Jensen's inequality applied to the function $x \log x$.
\end{proof}

\end{document}